\documentclass[12pt]{article}
\usepackage{amsmath}
\usepackage{amssymb}
\usepackage{graphicx}

\usepackage[a4paper,left=28mm,right=25mm,top=30mm,bottom=30mm]{geometry}

\RequirePackage[OT1]{fontenc}
\usepackage{dcolumn,warpcol}  
\RequirePackage{amsthm,amsmath,amssymb}
\usepackage{graphicx}\graphicspath{{pics/}}

\usepackage{latexsym}
\usepackage{arydshln,booktabs,bm}
\usepackage{array}

\usepackage{here}
\usepackage{bm} % bold math

\usepackage{lscape} % page rotation

\usepackage{natbib}
\usepackage{array}
\renewcommand{\baselinestretch}{1.2}

\renewcommand{\thefootnote}{\fnsymbol{footnote}}

\usepackage{fancyhdr}
\usepackage{bm} % bold math
\usepackage[latin1]{inputenc}
\usepackage{color}
\usepackage{graphicx}
\RequirePackage[OT1]{fontenc}
\usepackage{latexsym}
\usepackage{booktabs,rccol,flafter,arydshln,warpcol}
\usepackage{natbib}
\RequirePackage{amsthm,amsmath, natbib}

\usepackage{here}

\newcommand{\ccolhd}[1]{\multicolumn{1}{c}{#1}}

\newcommand{\Gcon}{$G^N_{\mathrm{con}}\,$}
\newcommand{\Gcov}{$G^N_{\mathrm{cov}}\,$}
\newcommand{\Gdag}{$G^N_{\mathrm{dag}}\,$}
\newcommand{\Greg}{$G^N_{\mathrm{reg}}\,$}

\newcommand{\Elin}{E_{\mathrm{\,lin}}}

\newcommand{\E}{{\it E}}

\newtheorem{prop}{Proposition}
\newtheorem{coro}{Corollary}
\newtheorem{defn}{Definition}
\newtheorem{lemma}{Lemma}

\newtheorem{theorem}{Theorem}
\newcommand{\ci}{\mbox{\protect{ $ \perp \hspace{-2.3ex}
\perp$ }}}

\newcommand{\n}[0]{\hspace*{.35em}}

\newcommand{\nn}[0]{\hspace*{.7em}}

\newcommand{\fla}{\mbox{$\hspace{.05em} \prec
\!\!\!\!\!\frac{\nn \nn}{\nn}$}}

\newcommand{\fra}{\mbox{$\hspace{.05em} \frac{\nn
\nn}{\nn
}\!\!\!\!\! \succ \! \hspace{.25ex}$}}

\newcommand{\dal}{\mbox{$  \frac{\n}{\n}
\frac{\; \,}{\;}  \frac{\n}{\n}$}}

\makeatletter
\renewcommand\section{\@startsection{section}{1}{\z@}%
{-3.25ex\@plus -1ex \@minus -.2ex}{1.5ex \@plus .2ex}%
{\normalfont\large\bfseries}}
\renewcommand\subsection{\@startsection{subsection}{2}{\z@}%
{-3.25ex\@plus -1ex \@minus -.2ex}%
{1.5ex \@plus .2ex}%
{\normalfont\normalsize\bfseries}}
\renewcommand\subsubsection{\@startsection{subsubsection}{3}{\z@}%
{-3.25ex\@plus -1ex \@minus -.2ex}%
{1.5ex \@plus .2ex}%
{\normalfont\normalsize\bfseries}}
\renewcommand\paragraph{\@startsection{paragraph}{4}{\parindent}%
{3.25ex \@plus1ex \@minus .2ex}%
{-1em}%
{\normalfont\normalsize\bfseries}}
\makeatother
\newcommand{\Gregone}{$G^N_{\mathrm{reg1}}\,$}
\newcommand{\Gregtwo}{$G^N_{\mathrm{reg2}}\,$}

\newcounter{repeat}
\newtheorem{repeatthm}[repeat]{Theorem}

\renewcommand{\baselinestretch}{1.2}

\begin{document}

\markright{}
\markboth{\hfill{\footnotesize\rm Kayvan Sadeghi and Nanny Wermuth
}\hfill}
{\hfill {\footnotesize\rm  Sequences of regressions} \hfill}
\renewcommand{\thefootnote}{}
$\ $\par
\fontsize{12}{14pt plus.8pt minus .6pt}\selectfont
\vspace{0.8pc}
\noindent{\Large \bf  Sequences of regressions and their independences\\[6mm]}
{\large  NANNY WERMUTH\\[2mm]}
{\it Department of Mathematics, Chalmers Technical University, Gothenburg, Sweden, and International Agency of Research on Cancer, Lyon, France; email: wermuth@chalmers.se\\[4mm]}
\noindent{\large KAYVAN SADEGHI\\[2mm]}
{\it Department of Statistics, University of Oxford, UK; email: kayvan.sadeghi@jesus.ox.ac.uk\\[6mm] }

\noindent{\bf ABSTRACT}: Ordered sequences of univariate or multivariate regressions provide
statistical models for analysing  data from randomized, possibly sequential  
interventions, from cohort or multi-wave panel studies, but also from 
cross-sectional  or retrospective studies.  Conditional independences are 
captured by what we name regression graphs, provided the generated distribution 
shares some properties with a  joint Gaussian distribution.  Regression graphs 
extend purely directed,  acyclic graphs by 
two types of undirected graph, one type   for   components of  joint responses 
and the other for  components of the  context vector variable. We review the special features and the history of regression 
graphs, prove   criteria  for Markov equivalence  and discuss the notion of a simpler statistical covering model. 
Knowledge 
of Markov equivalence provides alternative interpretations of a given 
sequence of regressions, is essential for machine learning strategies and 
permits to use the simple graphical  criteria of regression graphs 
on graphs for which the corresponding criteria are in general more complex.
Under the known conditions that   a Markov equivalent directed acyclic 
graph exists for any given regression graph,
we give a polynomial time algorithm to find one such graph. 
 \\[-1mm]

%Recursive sequences of univariate or multivariate regressions provide
%statistical models for analysing  data from randomized, possibly sequential  interventions, from cohort or multi-wave panel studies, but also from cross-sectional  or retrospective studies.  Conditional independences in the sequences are captured by what we
% name  regression graphs. Regression graphs are special types of  so-called chain graphs. They extend purely directed,  acyclic graphs by  adding undirected graphs  for the  relations of the components within a   joint response and within a context vector variable.
%Different  sets of  independences  may define different types of graph. These can nevertheless be  Markov equivalent,  that is imply the same set of independence statements. We  provide graphical criteria   to decide whether a given regression graph is Markov equivalent to another graph,  be it directed, undirected or mixed, and  give one polynomial time algorithm. Markov equivalence  provides alternative interpretations of  a given model and is essential for machine learning strategies.


\section{Introduction} \label{intro}

A common framework  to model, analyse and interpret  data for several, partially ordered  joint or single  responses
is a sequence of multivariate or univariate regressions where the responses may be continuous or discrete or of both types.  Each response is to be generated by  a set of its regressors, called its {\bf \em directly explanatory variables}. Based on prior knowledge or on statistical analysis, one is to decide which of  the variables in a set of potentially explanatory ones are  needed for the generating process.  Thus,  for each response, a first ordering determines what is potentially explanatory, named the past of the response, and  what  can never be directly explanatory, named the future.  Furthermore,  no variable is taken to be explanatory for itself.

Corresponding {\bf \em regression graphs}  consist of  {\bf \em nodes} and  of {\bf \em edges  coupling distinct nodes}.
The {\bf \em nodes  represent the variables} and the {\bf \em edges stand for conditional dependences},
directed or undirected. The directly explanatory variables for an individual response variable $Y_i$ show in the graph as the set of nodes from which
arrows start and point to node $i$. These nodes are  commonly named the {\bf \em parents of node  i}.

Every missing edge corresponds to a conditional  independence statement.  Edges are  {\bf \em  arrows  for
directed dependences}  and   {\bf \em lines  for undirected dependences} among {\bf \em  variables on equal standing}, that is among components of joint responses
or of context variables. Undirected dependences are often also called associations.
A given regression graph reflects a particular  type of study which may be a simple experiment,
a more complex sequence of interventions or an observational study.

  One of the common features  of pure
 experiments and  of sequences of  interventions with randomized, proportional  allocation of individuals to treatments,   is that,  by study design,  some  variables can be  regarded  to act just like  independent random variables. For instance, in an experiment
with proportional numbers of individuals assigned randomly to each  level combination of  several experimental conditions,  the set of explanatory variables contains no edge in  the corresponding regression graph, reflecting  a situation like mutual independence. Similarly, with fully randomized interventions, each treatment variable has  exclusively arrows starting from its node but no incoming arrow.  After statistical analysis, some conditional independences may be  appropriate additional simplifications which show as further missing
 edges.

Sequences of interventions  give a time ordering   for some of the variables. A time order  is also  present in cohort or  multi-wave panel studies and  in
retrospective studies which focus on investigating  effects of variables at  one  fixed time point in the past,  without the chance of  intervening. By contrast, in a strictly cross-sectional study, in which observations for all variables are  obtained at the same time, any particular  variable ordering is only assumed  rather than implied by  actual time.

The node set  is at the planning stage of empirical studies ordered into ordered sequences  of single or joint responses, $Y_a,$ $Y_b$, $Y_c\dots$ that we call {\bf \em blocks of variables on equal standing} and draw them  in figures as {\bf \em boxes}.  This determines for  the following statistical analyses that within each block there are undirected edges and between blocks there are directed edges,  the arrows. The first block  on the left  contains the {\bf \em primary responses} of  $Y_a$ and  the last block on the right contains  {\bf \em context variables, also named the background variables}. After statistical analyses,
arrows may start  from nodes within any block but  always end  at a node in one of the blocks in the future. Thus,  there are no arrows pointing to  context variables and all arrows point in the same direction, from right to left.  An {\bf \em intermediate variable} is a response to some variables
and  also  explanatory for other variables so that it has  both  incoming and  outgoing arrows in the regression graph.

As an  example, we take  data from a retrospective study with 283  adult females answering questions about their childhood when visiting their general practitioner, mostly for some minor health problems; see \cite{Hardt2008}.
A well-fitting graph is shown  in Figure \ref{childadv}. It contains two binary variables, $A,B$ and six quantitative variables.  Except for the directly recorded feature age in years, all other variables are derived  from answers to questionnaires,
coded so that high values correspond to high scores.

 The three blocks $a,b,c$ reflect here  a time-ordering of vector variables, $Y_a,Y_b,Y_c$ with  $Y_a$ representing  the joint response of primary interest, $Y_b$ an intermediate vector variable and $Y_c$ a context vector variable.
The  three individual components of  the primary response $Y_a$ are different aspects of how the respondent recollects aspects of her relationship to the mother. The intermediate variable $Y_b$ has two components that reflect severe distress during childhood. The three components of the context  variable $Y_c$ capture background information about the respondent and
about her family.

The graph of Figure \ref{childadv}, derived after  statistical analyses,  shows among other independences that $Y_a$ is conditionally independent of $Y_c$ given $Y_b$, written compactly in terms of sets of nodes as $\bm{a\ci c|b}$.  None of the components of
$Y_c$ has an arrow pointing  directly to a component of $Y_a$, but sequences of arrows   lead  indirectly from $c$ to $a$ via $b$.
\begin{figure} [H]
\centering
 \includegraphics[scale=.49]{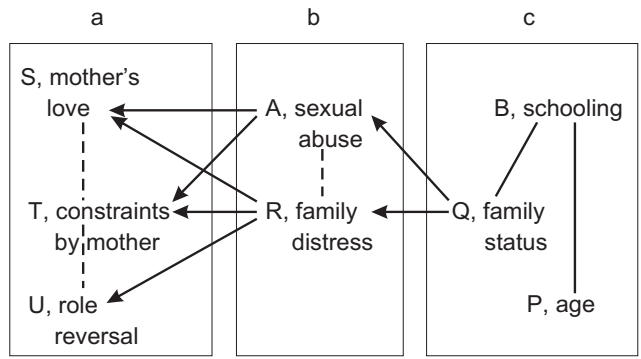}
 \caption{A well-fitting regression graph for data on $n=283$  adult females; within boxes are $Y_a, Y_b, Y_c$;  corresponding ordered partitioning of the node set on top of the boxes.} \label{childadv}
 \end{figure}
This  says, for instance,  that prediction of $Y_a$ is not improved by knowing the context variable $Y_c$ if information on the more recent  intermediate variable $Y_b$ is available. More  interpretations of the independences  are given later.
When some edges are missing and  each edge present corresponds to a substantial dependence,
 the graph  may also be viewed  as a research hypothesis on which variables are needed to generate the joint distribution;  see \cite{WerLau90}.  The goodness-of-fit  of such a hypothesis  can be tested  in future studies.

Two models are {\bf \em Markov equivalent} whenever   their associated graphs capture the same {\bf \em independence structure}, that is the graphs lead to  the same set of  implied  independence statements.   Markov equivalent  models  cannot be distinguished on the basis of statistical goodness-of-fit tests for any given set of data. This may pose a problem in machine learning contexts.  More precisely, knowledge about Markov equivalent models is essential for designing search
procedures that converge with an increasing  sample size to a true  generating graph; see \cite{CaKo03} for searches within the class of  {\bf \em directed acyclic graphs}, which consist exclusively of arrows and  capture independences  of ordered sequences in  single response  regressions.

\begin{figure} [H]
\centering    \includegraphics[scale=.49]{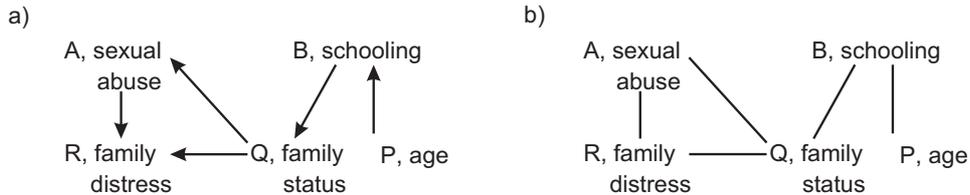}
 \caption[]{\small{Two Markov equivalent graphs to the one of $Y_b,Y_c$ of Figure \ref{childadv}.}}  \label{MEQch}
 \end{figure}

 More importantly though,  Markov equivalent models  may offer alternative interpretations of a given well-fitting model or open
the possibility of using different types of fitting algorithms.

As we shall see in Section 7, the graph for nodes $A,R,B,P,Q$  in blocks $b$ and $c$ of Figure  \ref{childadv} is Markov equivalent to both graphs of  Figure \ref{MEQch}.
 From knowing the Markov equivalence to  the graph in Figure  \ref{MEQch}a),
  the joint response model for $Y_b$ given $Y_a$ may also be fitted in terms of univariate regressions and  from the Markov equivalence to the graph in Figure  \ref{MEQch}b), one knows for instance  directly, using  Proposition 1 below, that  sexual abuse is independent of age and schooling given knowledge about family distress and family status.
%  each missing edge within $\{b,c\}$ of Figure \ref{childadv} can also be interpreted as conditional independence given all remaining five variables of $\{b,c\}$;  see Proposition 1 below.

Regression graphs are a subclass of {\bf \em the maximal ancestral graphs}  of  \cite{RichSpir02} and  these  are a subclass of  the {\bf \em summary graphs} of  \cite{Wer10}.  The two types  are called {\bf \em corresponding graphs} if they result after marginalising over a node set $m$ and conditioning on a disjoint node set  $c$ from a given  directed acyclic graph.  Both are  {\bf \em independence-preserving graphs}  in the sense that they give the  independence structure implied by the generating  graph for all the remaining nodes and further conditioning or marginalising can be carried out just as if the possibly much  larger generating graph were used.
The summary graph permits,  in addition, to trace possible distortions of generating dependences as they arise in conditional dependences among the remaining variables, for instance in parameters  of the maximal ancestral graph models.

In the following  Section 2, we introduce further  concepts and the notation needed to state at the  end of Section 2, some of  the main results of the paper
and  related results in the literature. In Section 3,  a well-fitting regression graph is derived for data of chronic pain patients.
%and give references  to previous results concerning regression graphs.
Sections 4, 5  and 6 may be skipped if one wants to  turn directly to formal definitions, new results  and proofs in Section 7. Section 4 reviews linear recursion relations  that are mimicked by graphs and lead to the standard and to  special ways of combining probability statements, summarized here in Section 5. In Section 6, some of the previous results in the literature for graphs and for Markov equivalences are highlighted. The Appendix contains details of the regressions analyses in Section 3.

\section{Some further concepts and notation}

Figure \ref{hypfigbef} shows  five ordered blocks,   to introduce the notion of connected  components of the graph to represent conditionally independent
responses given their common past.
\begin{figure} [H]
\centering
 \includegraphics[scale=.49]{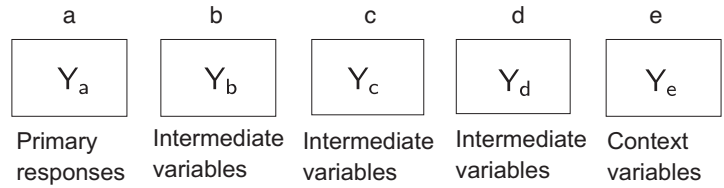}
 \caption[]{\small{A typical  first ordering: here of five vector variables, $Y_a, \ldots Y_e$; primary response $Y_a$ listed on the left, context variable $Y_e$ on the right, intermediate variables in between.}} \label{hypfigbef} \vspace{-2mm}
 \end{figure}
 \begin{figure} [H]
\centering
 \includegraphics[scale=.47]{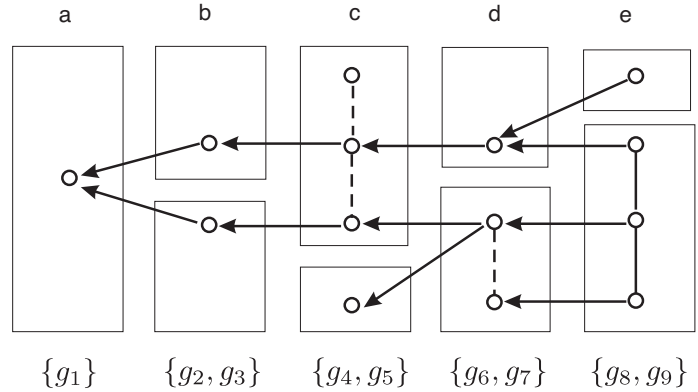}
 \caption[]{\small{A  regression graph for  14 variables  corresponding to  blocks $a$ to $e$ of Figure \ref{hypfigbef}.}}
  \label{hypfigaft} \vspace{-2mm}
 \end{figure}

 In the example of a  regression graph in Figure \ref{hypfigaft} corresponding to   Figure \ref{hypfigbef}, $Y_a$ is a single response, $Y_b$ has two component variables, both of $Y_c$ and $Y_e$  have four and $Y_d$ has three.
Each of the blocks $b$ to $e$ shows  two {\bf \em stacked boxes}, that is subsets of nodes  that are without  any  edge joining them. This is to indicate that disconnected  components of a given response  are conditionally independent given their past and that disconnected components of the context variables are completely independent.

 Graphs with dashed lines are {\bf \em covariance graphs denoted by  $\bm{G^N_{\mathrm{cov}}}$}, those with full lines are {\bf \em concentration graphs denoted by  $\bm{G^N_{\mathrm{con}}}$}; see \cite{WerCox98}. The names are to remind one
of their parametrisation in {\bf \em regular joint Gaussian distributions}, for  which the covariance  matrix is invertible and gives the {\bf \em  concentration matrix}.  A zero $ik$-element in   \Gcov means $i\ci k$ and a zero $ik$-element in  \Gcon means $i\ci k|\{1, \dots, d\}\setminus\{i,k\}$; see
 \cite{Wer76a} or \cite{CoxWer96},  Section 3.4.

  The regression graph  of  Figure \ref{hypfigaft} is consistent with the first ordering in Figure \ref{hypfigbef}  since no additional ordering is introduced, as it would have been  by arrows within blocks $a$ to $e$.  After statistical analysis, blocks of the first ordering
 are often subdivided into  the  connected components of the graph, $g_j$, shown here  in Figure \ref{hypfigaft} with  the help of the stacked boxes.
  For  several nodes in $g_j$,  each pair of nodes  $(i,k)$ is connected by at least one  undirected $ik$-path within $g_j$.  An {\bf \em  $\bm{ik}$-path} connects its endpoint nodes $i,k$  via a sequence of  edges coupling distinct other nodes along the path, named  {\bf \em the path's inner nodes}.

  For a regression graph,  \Greg, the node set $N$  has an ordered  partitioning into two subsets, $N=(u,v)$ distinguishing response nodes within $u$ from context nodes within $v$. The
   {\bf \em  connected components $\bm g_j$}, for $j=1, \dots J$, are  the disconnected, undirected  graphs  that remain after removing all arrows   from the graph.  Thus, the  displayed, stacked  boxes in Figure \ref{hypfigaft} are just a visual aid.  We say that  there is {\bf \em an edge  between subsets $\bm a$ and $\bm b$} of $N$ if there is an edge with one node in  $a$ and the other node in $b$. Then,  the subgraph induced by nodes $a\cup b$  {\bf \em is said to connected  in $a$ and $b$}.

   For  any one block  of stacked boxes,   different  orderings are possible.    We speak of a {\bf \em  compatible ordering} if each {\bf \em arrow}   starting  at a node in any $g_j$   points  to a node  in  $g_{<j}=g_1 \cup \dots \cup g_{j-1}$, but never to a node in $g_{>j}=g_{j+1}\cup \dots \cup g_{J}$, the {\bf \em past of $\bm{g_j}$}.

   {\bf \em   Full lines} are  edges  coupling context variables within $v$.
{\bf \em    Dashed lines} couple  joint responses within $u$. The regression graph is {\bf \em  complete} if every node pair is coupled. In this case, the statistical model is   {\bf \em saturated} as it is unconstrained
for some  given family of distributions.

  %Any full ordering of the connected components is  said to be compatible with   \Greg,

Let   $g_1, \dots g_J$ denote  any   compatible  ordering of the connected components
of  \Greg, then  a corresponding   joint density factorises as
\begin{equation}  f_N={\textstyle \prod _{j=1}^{J}} f_{g_j|g_{>j}},  \label{factdens}\end{equation}
into sequences regressions for the joint responses $g_j$ within $u$ and for separate concentration graph models
in disconnected $g_j$ within $v$.

In a {\bf \em generating process of $\bm{f_N}$ over a regression graph}, one starts with the density of  $g_{J}$ continues with the one of  $g_{J-1}$ given $g_{J}$ up to the density of $g_1$ given $g_{>1}$ so that  \eqref{factdens} is used for one given compatible ordering of the node set $N$.
Every  $ik$-edge present denotes  a  non-vanishing  conditional  dependence of $Y_i$ and $Y_k$ given some vector variable $Y_c$, written as $i\pitchfork k|c$ so that the graph is said to represent a {\bf \em dependence base} or to capture a dependence structure. The generating process attaches  the following  meaning to each $ik-$edge present in \Greg
 \begin{eqnarray} \label{pairw}& (i)&  \n  i\pitchfork k | g_{>j}   \quad \text{ for } i, k \text{ both in  a response component } g_j \text{ of } u \nonumber\\
 &(ii)& \n  i \pitchfork k| g_{>j}\setminus \{k\}  \nn  \text{ for } i \text{ in } g_j  \text{ of  }  u  \text{ and  } k \text{  in } g_{>j}\\
 &(iii) &\n     i\pitchfork k| v \setminus\{i,k\}  \text { for }  i, k \text{ both in  a context component }  g_{j} \text{ of  } v. \nonumber
 \end{eqnarray}

Notice that only for context variables, conditioning is on all other context variables  while for responses
conditioning is  exclusively on variables in their past.  When the dependence sign $\pitchfork$ is replaced
by the independence sign $\ci$, equations (2) give with missing edges for node pairs  $i,k$ the {\bf pairwise independence statements defining  the
independence structure of {\bm \Greg},} given the composition and the intersection property  discussed below.

An equivalent,  more compact description of the set of defining pairwise independences
and a proof of  equivalence of this {\bf \em pairwise Markov property} to the global Markov property has been given for the class of mixed loopless graphs,
which contain regression graphs as a subclass; see Sadeghi and Lauritzen (2011); see also  \cite{KangTian09}, \cite{PeaPaz87}, \cite{MarLup11} for    relevant, previous results.
 A {\bf \em global Markov property} permits to read off the graph all independence statements implied by the graph.

 Equation  \eqref{pairw}$(i)$ holds for the  conditional covariance graphs of  joint  responses
$g_j$ within $u$ having dashed lines as edges,   \eqref{pairw}$(ii)$ for the  dependences  of the single
responses within $g_j$ on variables in the past of $g_j$ having arrows as edges and equation  \eqref{pairw}$(iii)$   for the   concentration graph  of the context variables within $v$ having full lines as edges.
For instance, from the definition of the missing edges corresponding to \eqref{pairw}, one can  derive for Figure \ref{childadv},  $S\ci U|bc$  by \eqref{pairw}$(ii)$,   $P\ci Q|B$  by \eqref{pairw}$(iii)$, and both $A\ci B|PQ$ and  $A\ci P|BQ$ by  \eqref{pairw}$(i)$ using first principles and the two special properties of the generated distributions named
composition and intersection.

Notice that each missing edge of a regression graph corresponds to  an independence statement for the uncoupled node pair; see also Lemma \ref{lem:21} and Lemma \ref{lem:22} below. Therefore, regression graphs represent one special class of the so-called {\bf \em independence graphs}. Whenever
a regression graph \Greg consists  of {\bf \em two disconnected graphs}, for $Y_a$ and $Y_b$ say, since no path leads from a node in $a$ to a node in $b$, and $a\cup b=N$, then $a\ci b$ or $f_N=f_a f_b$, and the two vector variables
may be  analysed separately. Therefore,  we treat in Section 7 of this paper only connected   regression graphs.

All {\bf \em graphs discussed}  in this paper   {\bf \em have no loops}, that is no edge  connects a node to itself and  they have {\bf \em at most one edge between two different  nodes}.  Recall that an  $ik$-path in such a graph  can be described   by a sequence of its nodes.
By convention,   an $ik$-path without inner nodes is an edge. For every $ik$-edge,  the endpoints  differ, $i\neq k$.
An $ik$-path with $i=k$ has at least three nodes and is called a {\bf \em cycle}.

A three-node path  of arrows
may contain only one of  the  three types of inner nodes shown in  Figure \ref{dagV}, called {\bf \em transition,  source and sink node}, respectively.
 \begin{figure}[H]
 \centering
           {\includegraphics[scale=.47]{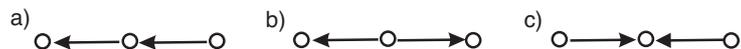}}
            \caption[]{\small{The three types of three-node paths in directed acyclic graphs with inner nodes named  a) transition, b) source, c) sink node (or in directed acyclic graphs: collision node).}}
  \label{dagV}      \end{figure}
 A  {\bf \em path  is  directed} if all its inner nodes are transition nodes. In a  {\bf \em directed  cycle}, all edges  are arrows
pointing   in the same direction and one returns to a starting node following the direction of the arrows. A regression graph contains  no directed cycle and no  {\bf \em semi-directed cycles}, which have at least one undirected edge in an otherwise directed cycle.  If  an arrow starts on a directed $ik$-path at  $k$ and points to  $i$ then  node $k$ has been named  an {\bf \em ancestor} of node $i$ and  node $i$ a  {\bf \em descendant} of node  $k$.

The {\bf \em subgraph induced by  a subset $\bm a$ } of the node set  $N$ consists of  the nodes within $a$ and of the edges present in the graph within $a$.
A special type of induced subgraph,  needed here, consisting of three nodes and two edges, is named a  {\bf \em {\sf V}-configuration} or just a
{\sf V}. Thus, a three-node path forms a $\sf V$  if the induced subgraph has two edges.

An {\bf \em $\bm{ik}$-path is  chordless}  if for each of its  three  consecutive nodes $(h,j,k)$, coupled by an $hj$- edge and and $jk$-edge,  there is no additional $hk$-edge present in the graph.
In a  {\bf \em chordless cycle}  of  four or more nodes, the subgraph induced by every consecutive three  nodes  forms  a $\sf V$ in the graph. An {\bf \em undirected graph is   chordal}  if it contains no chordless cycle  in four or more nodes.

In regression graphs, there may occur the   three types of {\bf \em collision {\sf V}s} of Figure \ref{collV}.  \begin{figure}[H] \centering
           {\includegraphics[scale=.47]{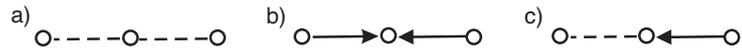}} \caption[]{\small{The three types of collision {\sf V}s in regression graphs: a) undirected, b) directed or sink-oriented, c) semi-directed;  for uncoupled path endpoints, the inner node is excluded from every independence statement that the graph implies for these endpoints.}}
\label{collV}
 \end{figure}
 Notice that in a directed acyclic graph, the only possible collision {\sf V} is directed and coincides with  the sink {\sf V} of Figure \ref{dagV}c).

An important  common feature of the three {\sf V}s of Figure \ref{collV} is that   the inner node is excluded from every  independence statements for its endpoints; see \eqref{pairw} and Lemma  \ref{lem:21}.
 In all other five possible types of {\sf V}-configurations of a regression graph, named {\bf \em transmitting {\sf V}s}, the inner node is instead included in the independence statement for the endpoints;  see \eqref{pairw} and  Lemma  \ref{lem:22} below. Notice that  for uncoupled endpoints, both paths a) and b) of Figure \ref{dagV} are transmitting {\sf V}s.
 Similarly, the definition of transmitting and collision nodes  remains unchanged if the {\sf V}s in Figure \ref{collV}
 are interpreted as $ik$-paths for which there may be an additional $ik$-edge  present in the graph.

 %Three more types of transmitting paths may occur in regression graphs,

A {\bf \em collision path}  has as inner nodes exclusively collision nodes, while a {\bf \em  transmitting path} has as inner nodes exclusively transmitting  nodes.  A chordless collision path in four nodes contains at least one dashed  line. In particular, it is impossible to replace all the edges  in such a four-node path by arrows and not generate at least one  transmitting {\sf V}. Thereby, the meaning of this  missing edge would  be changed and hence contradict its unique definition given from the generating process. The {\bf \em skeleton} of a graph results by replacing each edge present by a full line.
Now, two of the main new results of this paper can be  stated.

\begin{theorem}\label{thm:1}
Two  regression   graphs are Markov equivalent if and only if they have the same skeleton and the same sets of collision {\sf V}s,
irrespective of the type of edge.
\end{theorem}

%\begin{coro}
%A covariance graph and a  concentration graph with the same skeleton are  Markov equivalent if and only if  they contain no {\sf V}s.
%\end{coro}

\begin{theorem} \label{thm:2}
A regression graph with a chordal graph  for the context  variables  can be oriented to be   Markov equivalent to a directed acyclic graph in the same skeleton,  if and only if it does not contain any  chordless collision path in four nodes.
%every collision {\sf V}   in  \Greg  can be oriented as a sink {\sf V} and all collision {\sf V}s are preserved.
\end{theorem}

Sequences of regressions were  introduced and studied, without specifying  a concentration graph model for the context variables,  by \cite{CoxWer93}, \cite{WerCox04},  under the name of multivariate regression  chains,  reminding one of the  sequences of unconstrained models  that the class contains for Gaussian joint responses. An extension to graphs including a concentration
graph had already been proposed for directed acyclic graphs by \cite{KiiSpeCar84}. By this type of extension, the global Markov property of the  graph remains unchanged.

A criterion  for Markov equivalence of summary graphs has been derived by
\cite{Sadeghi09} who also shows that two different criteria  for maximal ancestral graphs are equivalent, those due to
\cite{ZhaZheLiu05} and  to \cite{AliRicSpi09}.  These available Markov equivalence results and the associated proofs increase  considerably in complexity, the larger the model class.
On the other hand, the Markov equivalence criterion of Theorem  \ref{thm:1} is simple and includes  as special cases all available equivalence results for directed acyclic graphs, for covariance graphs and for concentration graphs,  as set out in detail in Sections 6 and 7 here.

For  context variables taken as given, Gaussian regression graph models coincide with  a large subclass of  structural equation models  (SEMs), those permitting local modeling due to the factorisation property \eqref{factdens} and they are
without any  {\bf \em endogenous responses}.  Such responses have  residuals that are correlated with some of its regressors so that  the so-called endogeneity problem is generated,   by which,  for joint Gaussian distributions, a zero equation parameter need not correspond to any  conditional independence statement and a nonzero equation parameter is not a measure of conditional dependence. The consequence is  that ordinary least squares estimates of such equation parameters are typically strongly distorted.  This was recognized by \cite{Havelm43} who received a Nobel prize in economics for this insight in 1989.

For traditional uses of SEMs   see,  for instance,  \cite{Jor81}, \cite{Boll89}, \cite{Kline06}, while  \cite{Pea09} advocates SEMs as a framework  for causal inquiries. In the econometric literature forty years ago,
independences were always regarded as `overidentifying' constraints.

 For discrete
variables, more attractive features of regression graph models were derived  by \cite{Drton09}, who speaks of chain graph models of type IV
for multivariate regression chains in the case all variables on equal standing have covariance graphs. He   proves that each member in this class belongs to a curved exponential family, for a discussion of this notion see, for instance,  \cite{Cox06}, Section 6.8.
Discrete type IV models  form also a  subclass of marginal models; see   \cite{RudBerNem10}, \cite{BerRud02}.
Local independence statements that involve only variables in the past    are equivalent to more complex local independences
used by  \cite{Drton09}; see \cite{MarLup11}. These local definitions imply the pairwise independence formulation for missing edges corresponding to  equation \eqref{pairw} for any regression graph,  \Greg.

Two other types of chain graph have been studied as joint response models in statistics, the so-called
{\bf \em  AMP chain graphs} of \cite{AndMadPer97}, and the  {\bf \em LWF chain graphs} of \cite{LauWer89} and   \cite{Fryd90}.  They use the same factorisation as in equation \eqref{factdens}, but they are  suitable for modeling
data from intervention studies only when they  are Markov equivalent to a regression graph. The reason is  that
the conditioning set    for  pairwise independences  of responses includes in general other nodes within the same connected component. For AMP graphs,  the independence form of equation \eqref{pairw}  $(i)$   is replaced by
$$(i')  \nn \nn  i\ci k | g_{>j-1}\setminus{\{i,k}\}   \quad \text{ for } i, k \text{ both within a response component } g_{j} $$ while \eqref{pairw}  $(ii)$ and \eqref{pairw}  $(iii)$ remain unchanged.
For LWF graphs,  $(i)$ is also replaced   by $(i')$ and the independence form of $(ii)$ by
$$ (ii') \nn \nn  i \ci k| g_{>j-1}\setminus{\{i,k} \} \nn  \text{ for } i \text{ within a } g_j   \text{ and  } k \text{  in } g_{>j}.$$
As a consequence, each undirected subgraph in  an AMP chain graph is a concentration graph, and  an LWF chain
graph consists of sequences of concentration graphs. For the corresponding different types of parametrisations
of joint Gaussian distributions see \cite{WerWieCox06}.

Not yet systematically approached is the search for  {\bf \em covering models that capture most but not all
independences}  in a more complex graph but which may be easier to fit than the reduced  model;  see \cite{CoxWer90}. For regression graphs, details are explained here for a small example in Section 4, and
in Section 7,  first results are  given in Propositions \ref{MEQtoAMP} to \ref{MEQtoCON} and discussed using
Figures \ref{intersect} and \ref{indconc}.

Before we turn to  the different types of missing edges in more detail, we derive  a well-fitting regression graph for  data given by  \cite{Kappesser97}.

\section{Deriving and interpreting a regression graph}

 For  201 chronic pain patients, the role of the site of pain  during  a three week stay in a chronic
 pain clinic was to be examined.
  In this study, it was of main interest to investigate the  changes
  in two main symptoms before and after stationary treatment and  to understand  determinants of  the  overall  treatment success as
  rated by the patients, three months after they had left the clinic.
Figure \ref{figfopain} shows a first ordering of the variables derived in discussions between psychologists,
physicians and statisticians.

 The first ordering of the variables gives for each single or joint response a list of its possible explanatory variables,
 shown in  boxes to the right, but in Figure \ref{figfopain}  only those variables  are displayed that remained after  statistical analyses relevant for the responses of main interest.

 Selecting for each response all its directly explanatory variables from this list  and checking for remaining dependences among  components of joint responses, provides enough
 insight to derive  a  well-fitting regression graph model. With this type of local modeling, the
 reasons for the model choice are made transparent.

Of the available  background variables,  age, gender, marital status and others, only the binary variables,
 level of formal schooling (1:=less than ten years, 2:= ten or more years) and the number of previous illnesses in years
(min:=0, max:=16) are displayed in the far right box as the relevant context variables. The response of primary
interest, self-reported success of treatment, is listed in the box to the far left. It  is a score that
ranges between 0 and 35,  combining a  patient's  answers to a specific questionnaire.

 \begin{figure}[H]
\begin{center}
\includegraphics[scale=.49]{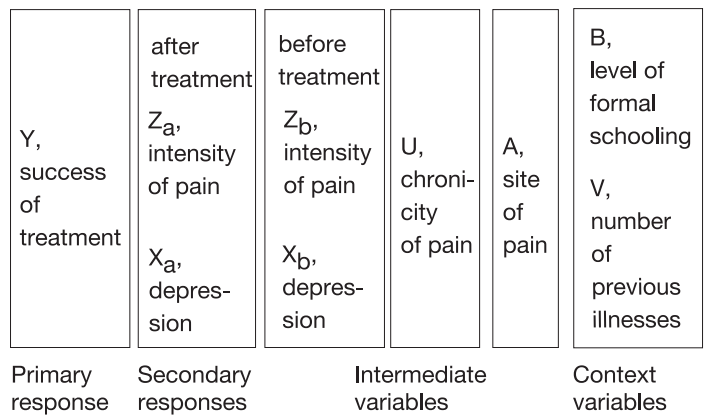}
\end{center}
\caption{\small First ordering of variables in the chronic pain study. There  are two joint responses, intensity of pain and depression. They   are the   main symptoms
 of chronic pain, measured here  before and after treatment. The components of each  response  are to be modeled conditionally
given the variables  listed  in boxes to their right.  }
\label{figfopain}
\end{figure}

There are a number of intermediate variables. These are  both explanatory for some variables and responses to others.
 Of these,  two are regarded as joint responses since they represent
 two symptoms  of a patient, intensity of pain  and depression.  Both are measured before treatment and directly after the three-week stationary stay.  Questionnaire scores are available  of depression
 (min:=0, max:=46) and
of the self-reported intensity of pain (min:=0, max:=10).
Chronicity of pain  is a score (min:=0, max:=8) that
 incorporates different aspects,  such as the frequency and duration of pain attacks,  the spreading of pain and the use of pain relievers.  
  In this study,  the  patients have  one of  two main sites of pain,  the pain is either
 on their  upper body, `head, face, or neck' or on their `back'.

 A well-fitting regression  graph is shown in  Figure \ref{figregpain}.
The   graph summarizes some important aspects of the results of the statistical analyses
for which details are  given in the Appendix.
In particular, it tells  which of the variables are directly explanatory, that is
which are important  for generating and predicting  a response,  by  showing arrows that start from each of these directly explanatory variables and
point to the response.

  \begin{figure}[H]
\begin{center}
\includegraphics[scale=.49]{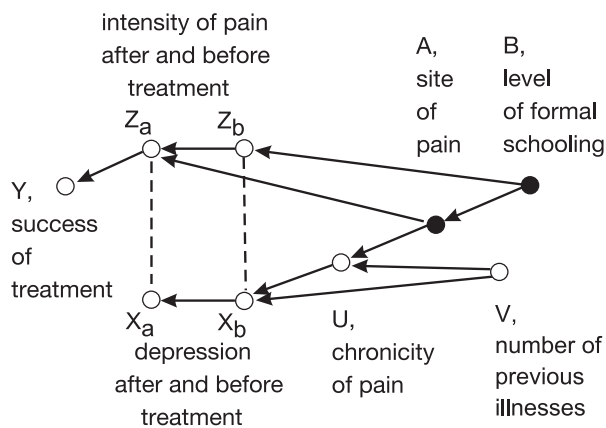}
\end{center}
\caption{\small Regression graph, well compatible with the data,  that results from the reported statistical analyses.  Discrete variables
are drawn as dots, continuous ones as circles.}
  \label{figregpain}
\end{figure}

Variables listed  to the right  of a response but without an arrow ending at this  response
do  not substantially improve the prediction of the response when used in addition to the directly
explanatory  variables. For instance, for treatment success,  only the pain intensity after the clinic stay is directly
explanatory and this pain intensity is an important  mediator (intermediate variable) between treatment success and site of pain.

 \begin{figure}[H]
\vspace{-4mm}
\begin{center}
\includegraphics[scale=.39]{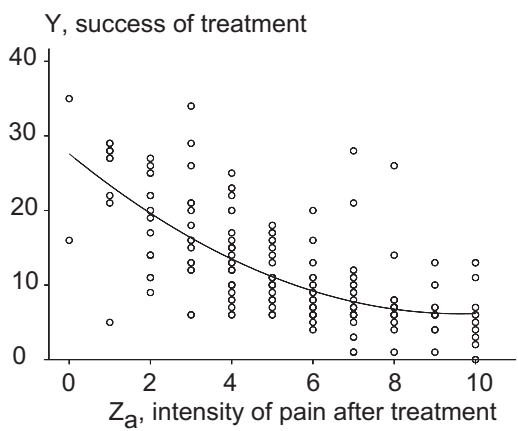}
\end{center}
\caption{\small Form of dependence of primary response $Y$ on $Z_{a}$.
}
\label{figpainquad}
\vspace{-4mm}
\end{figure}

 Scores of self-reported treatment  success are low  for almost  all patients with high pain scores after treatment that is for scores higher than 6; see Figure   \ref{figpainquad}.   Otherwise,
treatment success is typically judged to be higher  the lower the intensity of pain after treatment. This explains the nonlinear dependence of $Y$ on $Z_a$.

As mentioned before,  for back pain patients, the chronicity scores  are on average higher than for head-ache
patients and connected with a higher  chronicity of the pain are higher scores
of depression.
 These patients may possibly have tried too late,  after the acute pain had started,  to  get well focused help.
Both before and after treatment, highly depressed patients tend  to report higher intensities of pain
than  others.

%; ; see  also Figure \ref{figpainquad}.

The study provides no information  on which variables may explain these dependences between the symptoms that remain after having taken the available explanatory variables into account. However,  hidden  common explanatory  variables may exist in both  cases since these remaining dependences  between the symptoms do not
depend systematically on any other observed  variable.

Some  variables are  {\bf \em indirectly explanatory}.  An arrow starts from an indirectly explanatory variable,
  and points via a sequence of arrows   and intermediate variables to the response variable. For instance, the level of formal schooling
 and the site of pain  are both indirectly  explanatory for each of the symptoms  after treatment and  for the overall
 treatment success.

 Once the types and directions of the direct dependence are taken into account,  the regression graph
 helps to trace the development of chronic pain, starting from the context information on the level of schooling
 and the number of previous illnesses of a  patient. Thus,  patients with more
years of formal schooling are more likely to be chronic  head-ache patients. Patients  with a lower level of formal schooling
are more likely to be back-ache patients, possibly because more of them have jobs involving hard physical work.
 Back-ache patients  reach higher stages of  the chronicity of pain and report higher intensity of pain
 still after treatment and are therefore  typically less satisfied with the treatment they had received.

{\bf \em Graphical screening for nonlinear relations} and interactive effects (Cox and Wermuth, 1994) pointed to the nonlinear dependence of treatment
success  on intensity of pain after treatment but to no  other such relations.
The regression graph model is said to fit the data well because for each single response  separately, there is  no indication
that adding a further variable would substantially change the generated conditional dependences.  The seemingly unrelated dependences  of the symptoms
after treatment on those before treatment agree so well with the observations that they differ also little
from regressions computed separately, see the appropriate tables in the Appendix.

Had there been no nonlinear relation and no categorical variables as responses, the overall model fit  could also  have
been tested within the framework of  structural equation models once the  regression graph is available.
This graph is derived  here with the  local modeling steps that use the first ordering of the variables, just
in terms of univariate, multivariate and seemingly unrelated regressions. The regression graph  provides
a hypothesis that may be tested  locally and/or globally in future studies that include the same
 set of  nine variables.  In this case, no variable selection strategy would be used or needed.

 The available results for changes of the regression graph (Wermuth, 2011) that result after marginalising and conditioning
 provide a  solid basis for comparing the results of  any  sequence of regressions with studies that contain the same
 set of core variables but which have some of the variables omitted or which consider subpopulations,
 defined by levels or level combinations of other  variables. For instance for  comparisons with the current study,   the same chronicity score may not be recorded  in another  pain clinic or data may be available only for
  patients with pain in the upper body.

  {\bf \em The main substantive results of this empirical study} are that site of pain needs to be taken into account also in future studies since it is an important mediator between the intrinsic characteristics of a patient, measured here by the given context variables,  for both the overall treatment success and for the symptoms after treatment. For back-ache patients, the chronicity of pain and the depression score is higher than for the head-ache patients and the treatment is less successful since the intensity of pain remains high after the treatment  in the clinic.

In the following section we give three-variable examples of a Gaussian joint response regression and of  the three  subclasses
of regression graphs that have only one type of edge, of the covariance, the concentration and the directed acyclic graph to  discuss the different types of conditional dependences and the possible types of independence constraints associated with the corresponding regression graphs.

\section{Regressions, dependences and recursive relations}

For a quantitative response with linear dependences, the simple  regression model dates
back at least several centuries.  The fitting of
a least-squares regression line  had been developed separately by Carl Friedrich Gauss (1777--1855),
Adrien-Marie Legendre (1752--1833)   and Robert Adrain (1775 --1843).  The method extends  directly to models with several
explanatory variables.

The  most studied regression models are  for joint Gaussian distributions.
Regression graphs mimic important features of these linear models but represent also relations  in other distributions of continuous  and discrete variables, which  permit in particular nonlinear and interactive dependences.
In a  regular joint Gaussian distribution, %different parameters taking on value zero are equivalent to marginal or  to conditional independence.  L
let the mean-centered vector variable $Y$ have dimension three, then we write the covariance matrix, $\Sigma$, and the concentration matrix $\Sigma^{-1}$, with graphs shown in  Figure \ref{figcovcon3}, as
$$ \Sigma=\begin{pmatrix} \sigma_{11} & \sigma_{12} &\sigma_{13} \\
. & \sigma_{22 } &\sigma_{23} \\ .&.& \sigma_{33} \end{pmatrix}, \nn \Sigma^{-1}=\begin{pmatrix} \sigma^{11} & \sigma^{12} &\sigma^{13} \\
. & \sigma^{22 } &\sigma^{23} \\ .&.& \sigma^{33} \end{pmatrix},$$
where the  dot-notation indicates entries in a symmetric matrix. %The corresponding complete covariance and concentration graphs
%are shown in Figure \ref{figcovcon3}.
\begin{figure} [H] \vspace{-4mm}
\begin{center}
 \includegraphics[scale=.47]{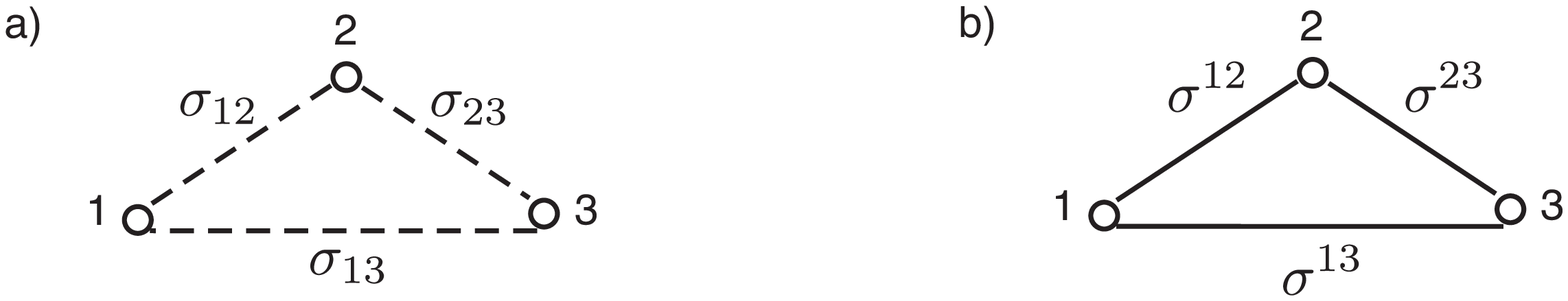}
 \caption[]{\small{For unconstrained trivariate Gaussian distributions, the parameters attached to the edges are those corresponding to  a) a covariance graph, b) a concentration graph. }} \vspace{-4mm}
        \label{figcovcon3}
 \end{center}
 \end{figure}
With  the edge of node pair  $(1,2)$ removed,  both graphs turn into a {\sf V} but have different interpretations.
The resulting independence constraints are for    Figures \ref{figcovcon3} a) and b), respectively,
$$  1\ci 2 \iff (\sigma_{12}=0) \nn \text{ and }  \nn  1\ci 2|3  \iff  (\sigma^{12}=0),$$
where the latter derives as an  important property of concentration matrices; for proofs see  \cite{CoxWer96},  Section 3.4 or \cite{WerCoxMar06}, Section 2.3.
For other distributions,  the independence interpretation of these two types of undirected   graph remains unchanged, but not the parametrisation. A similar statement holds for directed acyclic graphs and, more generally, for regression graphs.

For the linear equations that lead to a complete directed acyclic graph   for a trivariate  Gaussian distribution with mean zero, one  starts with three mutually independent Gaussian residuals $\varepsilon_i$ and takes the following
 system of equations, in which for instance $\beta_{1|3.2}$ is a regression coefficient for the dependence of response $Y_1$ on $Y_3$ when $Y_2$ is an additional regressor. Because of the form of the equations, one speaks of   triangular systems  also when the distribution of the residuals
is not Gaussian, but the residuals are just uncorrelated, or expressed equivalently, if each residual is uncorrelated with the regressors in its equation:
$$ Y_1=\beta_{1|2.3}Y_2+\beta_{1|3.2} Y_3+\varepsilon_1$$
\begin{equation}Y_2=\beta_{2|3} Y_3+\varepsilon_2  \label{triangeq3}\end{equation}
$$Y_3=\varepsilon_3.$$
 When the residuals do not follow Gaussian distributions, the probabilistic independence interpretation is lost,
but  the lack of a  linear relation can be inferred with any vanishing regression coefficient.

In econometrics,  Hermann Wold (1908--1992) introduced such systems  as linear recursive equations with uncorrelated residuals.
Harald Cram\'er  (1893--1985) used the term linear  least-squares equations for residuals  in a population being uncorrelated with the regressors  and
the notation for the regression coefficients is an adaption of the one introduced by Udny Yule  (1871--1951)   and William Cochran  (1909--1980).

 In  joint Gaussian distributions, independence constraints  on  triangular systems mean  vanishing equation parameters and missing edges in directed acyclic graphs, such as
$$  1\ci 2|3  \iff   (\beta_{1|2.3}=0) \nn \text{ and }  \nn  2 \ci 3 \iff  (\beta_{2|3}=0).$$ The  complete directed acyclic graph defined implicitly with  equations \eqref{triangeq3}  is displayed in Figure  \ref{figdagreg3}a).
\begin{figure} [H]
\begin{center}
 \includegraphics[scale=.45]{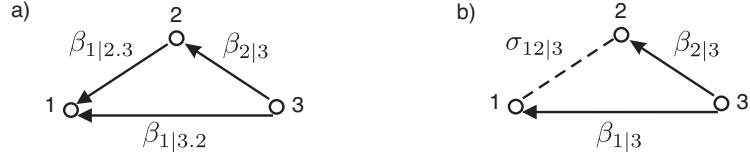}
 \caption[]{\small{Parameters of a  Gaussian distribution in: a) a complete \Gdag, b) a complete \Greg.}}
        \label{figdagreg3} \vspace{-3mm}
 \end{center}
 \end{figure}
For the smallest  joint response model with the complete graph shown in Figure \ref{figdagreg3}b), we take  both Gaussian variables $Y_1$ and $Y_2$ to depend on a Gaussian  variable $Y_3$, to get  equations \eqref{multreg3} with residuals having zero means and being uncorrelated with $Y_3$:
\begin{equation}  Y_1=\beta_{1|3}Y_3+ u_1, \nn \nn
Y_2=\beta_{2|3} Y_3+u_2,   \nn \nn Y_3=u_3.   \label{multreg3}\end{equation}
Here,  $\sigma_{12|3}=\E(u_1 u_2)$.  The generating processes and hence the interpretation differs for  the two
models in equations \eqref{triangeq3} and \eqref{multreg3}. In the corresponding graphs of Figures \ref{figdagreg3}a) and \ref{figdagreg3}b),  the vanishing of  the edges for pairs (1,2) and (2,3)
mean the same independence constraints  since
$$  1\ci 2|3  \iff   (\sigma_{12|3}=0)  \iff (\beta_{1|2.3}=0) \nn \text{ and }  \nn  2 \ci 3   \iff  ( \beta_{2|3}=0),$$
but the edges  for pair (1,3)  capture different  dependences, $1\pitchfork 3$ and $1\pitchfork 3|2$, respectively.
Again, taking away any edge generates a {\sf V}. Taking away any two edges means to combine  two independence statements.  This is discussed further in  the next section.

One of the  special important features of the
linear least-squares regressions is that the residuals  are uncorrelated with the regressors. The effect is that the model part
coincides with a conditional linear expectation as illustrated here with a model for response $Y_1$ and regressors $Y_2,Y_3$, which we take, as mentioned before,  as   measured in deviations from their means.  For instance, one gets for
 $$   Y_1=\beta_{1|2.3}Y_2+\beta_{1|3.2}Y_3+\varepsilon_1, $$
 \begin{equation}  \Elin(Y_1|Y_2,Y_3)=\beta_{1|2.3}Y_2+\beta_{1|3.2}Y_3 \n.\end{equation}

 There is a recursive relation for least-squares regression coefficients; see  \cite{Coch38}, \cite{CoxWer03}, \cite{MaXieGeng06}.
 It  shows for instance with
 \begin{equation} \beta_{1|3}=\beta_{1|3.2}+\beta_{1|2.3}\beta_{2|3} \label{rreg}\end{equation}
 that $\beta_{1|3.2}$, the partial coefficient of $Y_3$ given also $Y_2$ as a regressor for $Y_1$,  coincides with the marginal
 coefficient, $\beta_{1|3}$, if and only if $\beta_{1|2.3}=0$ or $\beta_{2|3}=0$.

The method of maximizing the likelihood was recommended by  Sir Ronald  Fisher  (1890--1962) as a general estimation technique  that applies also  to  regressions with  categorical or quantitative responses. One of the most attractive features of the method concerns properties of the estimates.
Given two  models with parameters
that are in  one-to-one correspondence, the same one-to-one transformation leads from   the maximum-likelihood estimates under one model to those of the other.

Different single response regressions, such as   logistic,  probit, or  linear regressions, were described as special cases of the generalized linear model by
\cite{NelWed72}; see also \cite{McCullNeld89}.  In all of these  regressions, the vanishing of the  coefficient(s)  of  a regressor  indicates conditional independence of the response given all directly explanatory  variables for this response.

The general linear model with a vector response, also called multivariate linear regression, has  identical sets of regressors  for each component variable of a response vector variable. Maximum-likelihood estimation of regression coefficients for a joint Gaussian distribution reduces to linear-least squares fitting for each component separately; see \cite{And58}, Chapter 8.

With different sets of regressors for the components of a vector response, seemingly unrelated regressions (SUR)    result
and   iterative methods are needed for estimation; see \cite{Zell62}.
 For small sample sizes,   a given solution of the likelihood equations of a Gaussian SUR model may not  be  unique; see \cite{DrtRich04},  \cite{Sund10}, while for exclusively discrete variables this will never happen; see \cite{Drton09}. For mixed variables, no corresponding results are available yet.

 In general, there  often exists {\bf \em a covering model with nice estimation properties.} For instance,  one of the  above described  Gaussian SUR models  that requires iterative fitting has regression graph  $$\circ \fra\circ \dal\circ\fla\circ \n,$$ A generating process starts with independent explanatory variables, each of which relates
 only  to one of the two response components, but these are correlated given both regressors. There is a simple covering model, in which two missing arrows  are added to the graph to obtain a general linear model. In that case,
 the new graph does not provide a dependence base, but closed form maximum-likelihood estimates are available.

For a vector variable of categorical responses only, the multivariate logistic regression of \cite{GlonMcCul95}  reduces to separate main effect logistic regressions for each component of the response vector provided that certain   higher-order interactions
vanish; see \cite{MarLup11}. In the context of structural equation models (SEMs), dependences of binary categorical variables are modeled in terms of probit regressions. These do not differ substantially from logistic regressions whenever
the   smallest and largest events occur at least with probability 0.1; see \cite{Cox66}.

Multivariate linear regressions as well as SUR models belong to
the framework of SEMs even though this general class had been developed in econometrics to deal appropriately
with endogenous responses. % defined by the existence of  correlations between the residuals and some regressors.
Estimation methods for SEMs were discussed in the  Berkeley  symposia on mathematical statistics  and probability from 1945 to 1965, but some identification issues  have been settled only recently; see
 \cite{FoyDraDrt11} and for relevant previous results  \cite{BriPea02},  \cite{StaWer05}.

In statistical models that treat all variables on equal standing, the variables are not assigned roles of responses or regressors  and undirected measures of dependence are used instead of coefficients of directed dependence. In the  concentration graph models, the undirected  dependences are conditional  given all remaining variables  on equal standing.

 For instance, for categorical variables, these models  are better  known as  graphical log-linear models; see \cite{Birch63},  \cite{Causs66}, \cite{Goodm70}, \cite{BisFieHol75},  \cite{Wer76a}, \cite{DaLauSpeed80}. For  Gaussian random variables,  these had been introduced as  covariance selection models; see \cite{Dem72}, \cite{Wer76b}, \cite{SpeedKiv86},   \cite{DrtonPer04}, and  for mixed variables as graphical models for conditional Gaussian (CG) distributions; see  \cite{LauWer89}, \cite{Edw00}.

For a mean-centered vector variable $Y$,  the elements of the covariance matrix $\Sigma$ are  $\sigma_{ij}=\E(Y_i Y_j)$.
If $\Sigma$ is invertible, the covariances  $\sigma_{ij}$ are in a one-to-one relation with  the concentrations $\sigma^{ij}$, the elements  of the concentration matrix $\Sigma^{-1}$. There is a recursive relation for concentrations; see \cite{Dem69}. For a trivariate distribution
\begin{equation} \sigma^{23.1}=\sigma^{23}-\sigma^{12}\sigma^{13}/\sigma^{11},  \label{rrcon}\end{equation}
where $\sigma^{23.1}$ denotes the concentration of $Y_2,Y_3$ in their  bivariate marginal distribution. Thus,  the overall concentration  $\sigma^{23}$ coincides with $\sigma^{23.1}$
if and only if $\sigma^{12}=0$ or $\sigma^{13}=0$.

Alternatively  in covariance graph models, the undirected measures  for variables on equal standing are   pairwise marginal dependences.  For Gaussian variables, these models had been introduced as   hypotheses
linear in covariances; see \cite{And73}, \cite{Kau96},    \cite{Kii87}, \cite{WerCoxMar06},  \cite{ChaDrtRich07}.  For categorical variables, covariance graph models  have been studied  only more recently; see \cite{DrtRich08}, \cite{LupMarBer09}. Again, no similar estimation results are available for general mixed variables yet.

There is also a recursive relation for covariances; see \cite{And58}, Section 2.5.
 It shows for instance, for  just three components  of $Y$ having a Gaussian distribution, with
\begin{equation}  \sigma_{12|3}=\sigma_{12}-\sigma_{13}\sigma_{23}/\sigma_{33},  \label{rrcov}\end{equation}
where $\sigma_{12|3}$ denotes  the  covariance of $Y_1,Y_2$ given $Y_3$.  Therefore, $\sigma_{12|3}$ coincides with $\sigma_{12}$
if and only if $\sigma_{13}=0$ or $\sigma_{23}=0$. By equations \eqref{rreg}, \eqref{rrcon}, \eqref{rrcov}, a
 unique independence statement  is associated with the endpoints of any {\sf V} in a trivariate Gaussian distribution.

In the context of multivariate  exponential families of distributions, concentrations are special canonical parameters and   covariances are special moment parameters with estimates of  canonical and moment parameters being asymptotically independent; see
\cite{Barn78}, page 122. Regression graphs  capture independence structures for more general types of distribution,  where operators
for transforming graphs mimic  operators for transforming  different  parametrisations of joint Gaussian distributions; see \cite{WerWieCox06},  \cite{WieWer10}, \cite{Wer10}.

 In particular,  by removing an edge from any $\sf V$ of a regression graph, one introduces an additional independence
 constraint just as in a regular  joint Gaussian distribution. For this, the generated distributions have to satisfy the composition and intersection property in addition to the general properties, as discussed in the next section.

\section{Using graphs to combine  independence statements}

 We now state the four standard properties of independences of any multivariate distribution; see e.g. \cite{Daw79}, \cite{Stu05}, as well as  two  special  properties of joint Gaussian distributions. The six taken together,  describe the combination and decomposition of  independences in  regression graphs, for instance  those resulting by removing  edges.  We  discuss when these six properties apply also to regression graph models.

Let $X, Y,Z$ be  random (vector) variables, continuous, discrete or mixed.  By using the same compact notation, $f_{XYZ}$  for a  given joint  density,  a probability
distribution or a mixture and by denoting the union of  say $X$ and $Y$ by $XY$, one has
 \begin{equation}  X\ci Y|Z   \iff    (f_{XYZ}=f_{XZ} f_{YZ}/f_{Z}), \label{eqci1} \end{equation}
 where for instance $f_{Z} $ denotes  the marginal density or probability distribution of $Z$. Since the order of listing variables for a given density is irrelevant, {\bf \em symmetry of conditional independence} is one of the standard properties,  that is

 $$(i)  \n  X\ci  Y|Z   \iff    Y\ci X|Z .$$
 Equation \eqref{eqci1} restated for instance for the   conditional distribution of $X$ given $Y$ and $Z$,
 $f_{X|YZ}=f_{XYZ}/f_{YZ}$,  is
\begin{equation}  X\ci Y|Z   \iff  (f_{X|YZ}=f_{X|Z}). \label{eqci2} \end{equation}

When two edges are removed from  a  graph in Figures \ref{figcovcon3}  and \ref{figdagreg3}, just one coupled pair remains, suggesting that the single node is  independent of the pair.

For instance  in  Figure \ref{figdagreg3}a),  with nodes $1,2,3$ corresponding in   this order to $X,Y,Z$,  removing the arrows for (1,2) and (2,3), leaves  (1,3) disconnected from
node 2. For any joint density,  implicitly
generated as  $f_{XYZ}=f_{X|YZ}f_{Y|Z}f_{Z} $, one has  equivalently,
$$ (X \ci Y|Z  \text{ and } Y\ci Z)  \iff   XZ\ci  Y. $$ In general, the {\bf \em contraction property} is  for $a,b,c,d$ disjoint subsets of $N$:
$$(ii) \n   (a\ci b|cd   \text{ and } b\ci c|d)  \iff ac\ci b|d.$$
It has become common to say that a {\bf \em distribution is generated over a given \Gdag} if the distribution factorizes as specified by the graph for  any  compatible ordering. For instance, for a trivariate distribution generated over the collision  ${\sf V}$ of Figure \ref{figdagreg3}b) obtained by removing the edge for (2,3), both orders $(1,2,3)$ and $(1,3,2)$ are compatible with the graph and
$f_{XYZ}=f_{X|YZ}f_{Y}f_{Z} $.

Conversely, suppose that $XZ \ci Y$ holds, then this implies  $X \ci Y$ and $Z \ci Y$ so that for instance  the same two
 edges as in Figure \ref{figdagreg3}b) are missing  in the corresponding covariance graph of Figure \ref{figcovcon3}a). In general, the {\bf \em decomposition property} is for $a,b,c,d$ disjoint subsets of $N$:
$$(iii) \n   a\ci bc|d     \implies (a\ci b|d \text{ and } a\ci c|d).$$

In addition,  $XZ \ci Y$  implies   $X \ci Y|Z$ and $Z \ci Y|X$ so that for instance
 the same two  edges  as in Figure \ref{figdagreg3}a)  are missing in the corresponding concentration graph of Figure \ref{figcovcon3}b).  In general, the {\bf \em weak union property} is for $a,b,c,d$ disjoint subsets of $N$:
$$(iv) \n   a\ci bc|d     \implies (a\ci b|cd \text{ and } a\ci c|bd).$$
Under some regularity conditions, all joint distributions  share     the four
properties $(i)$ to $(iv)$.

Joint distributions, for which the reverse implication of the decomposition  property $(iii)$ and of the weak union property $(iv)$ hold such as a regular joint Gaussian distribution,
are said to have, respectively,  the {\bf \em composition property} $(v)$ and the {\bf \em intersection property}  $(vi)$, that is for $a,b,c,d$ disjoint subsets of $N$:
 $$(v) \n   (a\ci b|d \text{ and } a\ci c|d) \implies  a\ci bc|d,$$
$$(vi) \n   (a\ci b|cd \text{ and } a\ci c|bd)  \implies a\ci bc|d.$$

The standard graph theoretical separation criterion has different consequences for
the two types of undirected graph corresponding for Gaussian distributions to concentration and to
covariance matrices. We say  {\bf \em a   path intersects  subset set $\bm c$} of node set $N$  if it has an inner node in $c$ and let $\{a, b, c, m\}$ partition  $N$   to  formulate known Markov properties. The notation is to remind one that with any independence statement $a\ci b|c$, one implicitly has marginalised over the remaining nodes in $m=V\setminus\{a\cup b\cup c\}$, i.e.\ one considers the marginal joint distribution of $Y_a,Y_b, Y_c$.
\begin{prop}\label{prop:31}{\em \cite{Lau96}}.
A  concentration graph, \Gcon, implies  $a\ci b|c$ if and only if every path  from $a$ to $b$ intersects $c$.
\end{prop}
\begin{prop}\label{prop:32}{\em \cite{Kau96}.}
A covariance graph, \Gcov,  implies  $a\ci b|c$ if and only if every path  from $a$ to $b$ intersects $m$.
\end{prop}
Notice that  Proposition \ref{prop:31}  requires the intersection property, otherwise one could not conclude for three distinct nodes $h,i,k$
e.g. that ($h\ci i|k$ and $h\ci k|i$)  implies $h\ci ik$ while
Proposition \ref{prop:32}  requires the composition property, otherwise one could  conclude  e.g.
that ($h\ci i$ and $h\ci k$) implies   $h\ci ik$. %The same holds if a subset $c$ of $N\setminus\{h,i,k\}$ is added to the conditioning set.

%A  subgraph induced by nodes $a\cup b$ in \Gcov  is the covariance graph  $G^{a\cup b}_{\rm cov}$
%and the  % By the definition of an overall covariance graph \Gcov, an $ik$-edge is absent if and only if  the graph implies $i\ci k$. Thus, edges  from $a\cup b$  to nodes outside of $a\cup b$  in \Gcov do not affect $G^{a\cup b}_{\rm cov}$, the covariance graph of $Y_{a\cup b}$.
\begin{coro}\label{ImplicCov}
A covariance graph, \Gcov, or a concentration graph,  \Gcon,  implies $a\ci b$  if and only if in  the subgraph induced by $a \cup b$, there is no edge between  $a$ and $b$.
\end{coro}

%It can be shown that the independence structure  of a regression graph is fully specified by  the  pairwise independences \eqref{pairw} of  each missing edge
%if both properties $(v)$ and  $(vi)$  hold in addition to the standard ones;
\begin{coro}\label{indep}
 A regression graph, \Greg, captures an  independence structure for  a  distribution with  density $f_N$ factorizing as \eqref{factdens}  if  the composition  and intersection property  hold for $f_N$, in addition to the standard properties of each density.
\end{coro}
\begin{proof}  Given the intersection property $(vi)$, any node $i$ with missing edges to nodes $k,l$ in a concentration graph of node set $N$ implies $i\ci \{k,l\}| N\setminus\{i,k,l\}$ and given the composition property $(v)$, any node $i$ with missing edges to nodes $k,l$ in a covariance  graph given $Y_c$  implies $i\ci \{k,l\}| c$. % The same type of argument applies if  the two edges are removed in triangles of a larger regression graph  since then  only the conditioning set of the corresponding independence statements
% is  enlarged by some fixed subset of  other nodes of the graph.
\end{proof}

For purely discrete and for Gaussian distributions, necessary and sufficient conditions for the intersection property $(vi)$ to hold are known; see \cite{SanMMouRol05}. Too strong sufficient
conditions are  for  joint Gaussian distributions that they are regular and for  discrete variables,  that the  probabilities are strictly positive.

The  composition property $(v)$ is satisfied in  Gaussian distributions and for
 triangular binary distributions with at most main effects in symmetric $(-1, 1)$ variables; see
\cite{WerMarCox09}.

  Both properties $(v)$ and $(vi)$ hold, whenever  a distribution may   have  been generated over a possibly larger  parent graph;  see \cite{Wer10}, \cite{MarWer09}, \cite{WerWieCox06}.
{\bf \em Parent graphs} are directed acyclic graphs that do not only capture an independence structure but are also  a dependence base with a  unique  independence statement assigned to each {\sf V} of the graph.  {\bf A distribution generated over a parent graph} mimics these properties of the parent graph.

It is known that
every regression graph can be generated by  a larger directed acyclic graph but not necessarily every statistical regression graph model can be generated in this way; see \cite{RichSpir02}, Sections 6 and 8.6.

One needs similar  properties  for distributions generated  over  a regression graph. {\em \bf A graph is edge-minimal
for the generated distribution} if  the distribution has a pairwise independence for each edge missing and a non-vanishing dependence for each edge present in the graph. For the generated distribution to have a unique independence
statement assigned to each missing edge, it  has to be
  {\bf \em singleton transitive} that is, for  $h,i,k,l$ distinct nodes of  $N$,
%\begin{equation}
$$ (i\ci k|l   \text{ and } i\ci k| lh) \implies (i \ci h| l \text{ or } k\ci h | l).$$
this says, that  in order to have both a conditional independence of $Y_i, Y_k$ given $Y_l$ and given $Y_l, Y_h$, there has
to be at least one additional independence involving the variable $Y_h$,  the additional variable  in the conditioning set.
%\label{WTprop}\end{equation}
For  graphs representing a dependence structure, this can be expressed equivalently, as
$$ (i  \pitchfork h|l   \text{ and } k \pitchfork h| l  \text{ and } i\ci k|l) \implies  i   \pitchfork  k | \{l,h\}$$
and
$$ (i \pitchfork  h|l   \text{ and }   k \pitchfork  h| l  \text{ and } i\ci k|\{l,h\}) \implies  i  \pitchfork  k |  l ,
$$ which says that  in the  distribution  there is a  unique independence statement that corresponds  to each {\sf V} in the graph.
For a $2\times 2\times 3$ contingency table,  an example  violating singleton-transitivity has been given with equation (5.4) by \cite{Birch63}.

There exist these  peculiar  types of
incomplete families of distributions; see \cite{LehSch55}, \cite{Brown86},  \cite{ManRue87},  in which   independence statements  connected with  a  $\sf V$  may  have the inner node both within  and outside the conditioning set; see  \cite{WerCox04}, Section  7, \cite{Darr62}.
Such independences have also  been characterized as being not representable in joint Gaussian distributions; see \cite{LnenMatus07}. These distributions  and those that are faithful to graphs are of  limited interest in application in which one wants to interprete sequences of regressions.

 {\bf \em Distribution are said to be  faithful to a graph} if every  of its independence constraints  is  captured by a given independence graph; see \cite{SpiGlySch93}. As is proven in a forthcoming paper,
this requires for regression graphs that (1) the graph represents  both an independence and a dependence structure, and that  (2)  the distribution satisfies the composition and the intersection property and is
{\bf \em weakly transitive}, a property that is the following extension of singleton transitivity for  node $h$
replaced by a subset  $d$ of $N\setminus\{i,k,l\}$ that may contain several  nodes: $$ (i\ci k|l   \text{ and } i\ci k| \{l,d\}) \implies (i \ci d| l \text{ or } k\ci d| l).$$

% and (3)   dependences introduced by several active paths  for any pair $(i,k)$ do not cancel.
%An active path is specified here for regression graphs  in Definition 1 and for  directed acyclic graphs, it has  been called d-connecting in the literature.

This faithfulness property imposes strange constraints on parameters whenever  more than two nodes induce  a complete subgraph in the graph; see for instance Figure 1 in \cite{WerMarCox09} for three binary variables.
An early example of a regular Gaussian distribution that does not satisfy weak transitivity  is due to \cite{CoxWer93},  equation (8).

Notice that in general, the extension of singleton transitivity to weak transitivity excludes parametric cancelations  that  result from several paths
connecting the same node pair. This  the only type of a possible parametric cancelation  in regular Gaussian distributions; see \cite{WerCox98}.

However, the constraints are mild for  distributions corresponding to regression graphs that form a dependence base and that  are forests.  {\bf \em Forests} are the union of disjoint trees and
 a {\bf \em tree} is  a connected undirected graph with one unique path  joining every node pair.
 \begin{lemma} A positive distribution is faithful to a forest  representing both a an independence and a dependence structure if it is singleton transitive.
 \end{lemma}
 \begin{proof} Positive distributions satisfy the intersection property and for concentration graphs, the composition property is irrelevant, Given the above characterizations of faithfulness and of weak transitivity, there are in a forest no cancelations due to several paths  connecting the same node pair. Hence, weak transitivity will be violated only if the singleton transitivity fails.
 \end{proof}

 \begin{coro} A regular Gaussian distribution  is faithful to a forest  representing both an independence and a dependence structure.
 \end{coro}

 Notice that forests include trees and Markov chains as special cases. If they form dependence bases they are Markov equivalent to very special types of parent graphs but they
 are rarely of interest in statistics when studying sequences of regressions.

 \section{Some early results on graphs and Markov equivalence}

  In the past,  results concerning graphs and  Markov equivalence have been obtained  quite independently in
 the mathematical literature on characterizing different types of graph, in the statistical  literature on
 specifying types of multivariate  statistical models,  and in the computer science literature on
 deciding on special properties of  a given graph or on designing fast algorithms for transforming  graphs.

 For instance, following the simple enumeration result for {\bf \em labeled trees} in $d$ nodes, $d^{d-2}$,  by Karl-Wilhelm Borchardt   (1817-1880),  it could be shown that these trees are in  one-to-one correspondence to  distinct strings of size $d-2$; see \cite{Cay889}.  Much later,  labeled trees were recognized to form the subclass of directed acyclic graphs
 with exclusively  source {\sf V}s  and  therefore to be  also Markov equivalent to chordal concentration graphs that are without chordless paths in four nodes; see \cite{CasSie03}.

In the literature on graphical Markov models, a  number of different names have been in use for  a sink {\sf V,} for instance `two arrows meeting head-on'  by \cite{Pea88},  `unshielded collider'  by \cite{RichSpir02}, and  `Wermuth-configuration'   by  \cite{Whit90},   after it had been recognized that, for Gaussian distributions, the parameters of  a  directed acyclic graph model without sink {\sf V}s are
in one-to-one correspondence to the parameters in its skeleton concentration graph model.

\begin{prop}\label{prop40} {\em \citep{Wer80}, \citep{WerLau83}, \citep{Fryd90}.}
A directed acyclic graph is  Markov equivalent  to a concentration graph of the same skeleton if and only if  it has no  collision $\sf V$.
\end{prop}

Efficient algorithms to decide whether an undirected graph can be oriented  into
a directed acyclic graph, became  available  in the computer science literature under the name of  perfect elimination schemes; see \cite{TarYan84}.
When algorithms were designed later to decide which arrows may be flipped in a given \Gdag, keeping  the same skeleton and  the same set of sink {\sf V}s,  to get to a list of all Markov equivalent \Gdag s, these early results by Tarjan and Yanakakis  appear are not referred to directly; see \cite{Chi95}.

The number of equivalent characterizations of concentration graphs that have perfect elimination schemes has increased steadily, since they were introduced   as rigid circuit graphs by \cite{Dir61}. These graphs are not only   named  `chordal graphs',  but also `triangulated graphs'', `graphs with the running intersection property'  or  `graphs with only complete prime graph separators'; see \cite{CoxWer99}.

By contrast, for a covariance graph that can be oriented to be Markov equivalent to a \Gdag of the  same skeleton,  chordless  paths  are relevant.
\begin{prop}\label{cono}{\em \citep{PeaWer94}.}
A covariance graph with a chordless  path in four nodes is not Markov equivalent to a directed acyclic graph in the same node set.
\end{prop}

For distributions generated over directed acyclic graphs,  sink {\sf V}s are needed again.

\begin{prop}\label{prop:40} {\em \citep{Fryd90},  \citep{VerPea90}.}
Directed acyclic graphs of the same skeleton are Markov equivalent if and only if they have the same sink {\sf V}s.
\label{MEQdag}\end{prop}

Markov equivalence of a concentration graph and a covariance graph model  is for regular joint Gaussian distributions equivalent to {\bf \em parameter equivalence}, which means that there is a one-to-one relation between the  two sets parameters. Therefore, an early  result on parameter equivalence for joint Gaussian distributions  implies the following Markov equivalence result for distributions satisfying both the composition and the intersection property.
%Notice that in exponential families, concentrations are canonical parameters.
\begin{prop}\label{conc}{\em \citep{Jen88}, \citep{DrtonRich08}.}
A  covariance graph is Markov equivalent to a  concentration  graph   if and only if  both
 consist of  the same complete, disconnected  subgraphs. \end{prop}

Fast  ways of  inserting an edge for every transition {\sf V}, of deciding on connectivity and on blocking flows  have been available in the corresponding Russian literature since 1970; see \cite{Din06}, but these results appear to have not  not been exploited  for the so-called lattice conditional independence models, recognized   as  distributions generated over $G^{N}_{\rm dag}$s without any transition {\sf V}s by \cite{AndMadPerTr97}.

Markov equivalence of other than multivariate regression chain graphs,   have been given by \cite{Rov05},
\cite{AndPer06} and \cite{RovStud06}.

With the so-called  global Markov property of a graph in node set $N$ and any disjoint subsets
$a,b,c$ of $N$, one can decide  whether the graph implies $a\ci b |c$. To give this property for a regression graph,
we use special types of path that have been called active;  see  \cite{Wer10}.  For this, let again $\{a,b,c,m\}$  partition the node set $N$ of \Greg.
\begin{defn}\label{def:33}
{\bf \em A path from $\bm  a$ to $\bm b$ in  $\bm{Greg}$ is active given $\bm c$} if its inner collision nodes are in $c$ or have a
descendant in $c$ and its inner transmitting nodes are in $m=N\setminus (a\cup b \cup c)$. Otherwise, {\bf \em  the path is said to break  given $\bm c$}
or, equivalently,  to break with $m$. \end{defn}
 Thus,  a path breaks  when $c$ includes an inner  transmitting node or  when $m$ includes an inner  collision node and all its descendants; see also Figure 4 of \cite{MarWer09}.

For directed acyclic graphs, an active path of Definition 1 reduces to the  d-connecting path of \cite{GeiVerPea90}.
Similarly, the following proposition coincides in that special case with those of  their so-called  d-separation.
Let   node set $N$ of \Greg be partitioned as above by $\{a,b,c,m\}$.
\begin{prop}\label{prop:1} {\em  \citep{CoxWer96},  \citep{Sadeghi09}.}
A regression graph, \Greg,  implies $a\ci b| c$ if and only if every  path between $a$ and $b$ breaks given $c$.
\label{globalreg}
\end{prop}

Thus, whenever \Greg implies   $a\ci b|c$, this independence  statement holds in the  corresponding sequence of regressions
for which the density $f_N$  factorizes as \eqref{factdens},  provided  that $f_N$ satisfies the  same properties of independences, $(i) $ to $(vi)$ of Section 5,  just like  a regular Gaussian joint density.
For example, in the graphs of Figure \ref{fig:2ex0}, node  $2$ is an ancestor of node $1$ so that \Greg  does not imply $3\ci 4|2$.

\begin{figure}[H] \centering
           {\includegraphics[scale=.47]{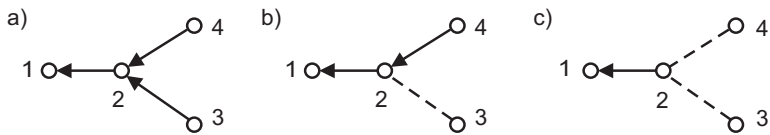}}
            \caption[]{\small{Three regression graphs, which imply $3\ci 4$ but not  $3 \ci 4|1 $.}}
        \label{fig:2ex0}
\end{figure}

Since covariance and concentration graphs consist only of one type of edge, the restricted versions in Propositions 1 and 2  of the defined path can be used for  their global Markov property.

\section{The main new results and proofs}

We now treat connected regression graphs in node set $N$ and corresponding distributions defined by sequences of regressions with joint discrete or continuous  responses, ordered in  connected components  $g_1, \ldots, g_r$ of the graph, and with
context variables in connected components, $g_{r+1}, \ldots, g_J$, which factorize
as in \eqref{factdens},  satisfy the pairwise independences of \eqref{pairw} as well as properties of  independence statements,
given as $(i)$ to $(vi)$ in  Section 5.

For the main result of Markov equivalence for regression graphs, we consider
 distinct nodes $i$ and $k$, node subsets $c$  of $N\setminus\{i,k\}$ and the notion of shortest active paths. \begin{defn}\label{def:shortestpi}
An  $ik$-path in \Greg  is a shortest active path $\pi$ with respect to $c$
if  every  $ik$-path  of  \Greg with fewer inner nodes breaks given $c$.
 \end{defn}
 Every chordless $\pi$  is such a shortest  path.  If the consecutive nodes  $(k_{n-1}, k_n, k_{n+1})$  on $\pi=(i= k_0,k_1,\dots,k_m=k)$
 induce a complete subgraph in \Greg,  we say that there is {\bf \em a triangle on the path}. In Figure \ref{shortest}a)  nodes 2,3,4
 form a triangle on the path $(1,2,4,3,5)$.

\begin{figure}[H] \centering
           {\includegraphics[scale=.47]{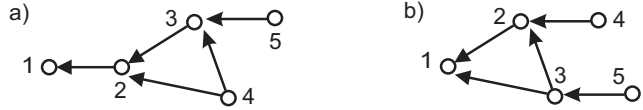}}
            \caption[]{\small{Graphs of active five-node  paths a) with path $ (1,2,4,3,5)$ the shortest active path, where $3$ is in $c$, b) active path $(4,2,1,3,5) $,  where $1$ is in $c$, and  a shorter active path $(4,2,3,5)$.}}
        \label{shortest}
\end{figure}

  If this path is an active path connecting the uncoupled node pair (1,5), then nodes 2 and 4 are  inner transmitting nodes outside $c$ and  the inner collision node 3 is in $c$. This path is then also the shortest active path connecting (1,5). The shorter path
 $(1,2,3,5)$ has   nodes 2 and 3 as inner transmitting nodes,  but is inactive since node 3 is in  $c$.

 By contrast  in Figure \ref{shortest}b), when path $(4,2,1,3,5)$ is an active path connecting the uncoupled node pair  (4,5),  then
 path (4,2,3,5) is a shorter active path.  To see this, notice that on an active  $(4,2,1,3,5)$ path, the inner collision node 1 is in $c$ and the inner transmitting nodes 2 and 3 are outside $c$. In  this case, the inner collision node 2   on the path $(4,2,3,5)$  has node 1 as a descendant in $c$, so that  this shorter path is also active.

We also use the following
results for proving Theorem \ref{thm:1}. The first two are direct consequences of Proposition \ref{prop:1} and imply the pairwise independences of equation \eqref{pairw}. Lemma \ref{lem:1}  results with  the independence form of \eqref{pairw}. Let $h,i,k$ be distinct nodes of $N$.
\begin{lemma}\label{lem:21}
For  $(h, i, k)$ a collision V in \Greg, the inner node $i$ is excluded from $c$ in every
independence statement for  $h,k$  implied by \Greg.
\end{lemma}
\begin{lemma}\label{lem:22}
For  $(h, i, k)$ a transmitting V in \Greg, the inner node $i$ is included in $c$ in every
independence statement for  $h,k$ implied by \Greg.
\end{lemma}
\begin{lemma}\label{lem:1}
A missing $ik$-edge in \Greg implies at least one independence statement $i\ci k|c$ for
$c$ a subset of $N\setminus\{i,k\}$.
\end{lemma}
We can now derive the first of the main new results in this paper.
\setcounter{repeat}{0}
\begin{repeatthm}
Two  regression   graphs are Markov equivalent if and only if they have the same skeleton and the same sets of collision {\sf V}s,
irrespective of the type of edge.
\end{repeatthm}
\begin{proof}
Regression graphs \Gregone and \Gregtwo are Markov equivalent if and only if for every disjoint subsets $a$, $b$, and $c$ of the node set of $N$,  where only $c$ can be empty,
\begin{equation}
 \text{(\Gregone}   \implies  a\ci b| c) \iff  \text{(\Gregtwo} \implies a\ci b| c) .  \label{p}
\end{equation}

Suppose first  that (\ref{p}) holds. By Lemma \ref{lem:1},  \Gregone and \Gregtwo have the same skeleton, and by Lemma \ref{lem:21} and Lemma  \ref{lem:22},  \Gregone and \Gregtwo have the same collision {\sf V}s.

Suppose next that  \Gregone and \Gregtwo have the same skeleton and  the same collision  {\sf V}s and
consider two arbitrary  distinct nodes $i$ and $k$ and any node subset $c$  of $N\setminus\{i,k\}$.
By Proposition \ref{prop:1},  (\ref{p})   is equivalent to stating that
 for every  uncoupled node pair $i, k$, there is an active
 path with respect to $c$ in \Gregone if and only if there is an active $ik$-path with respect to $c$ in \Gregtwo.

 Suppose further that path $\pi$  is  in \Gregone a shortest active $ik$-path with respect to $c$. Since \Gregone and \Gregtwo   have the same skeleton, the
path $\pi$ exists in \Gregtwo. We need to show that it is active. If all consecutive two-edge-subpaths of $\pi$
  are {\sf V}s then  $\pi$ is active in \Gregtwo.
Therefore, suppose that nodes $(k_{n-1},k_n,k_{n+1})$ on $\pi$  form a triangle instead of a $\sf V$.
It may be checked first, that in all other possible triangles  in regression graphs that can appear on $\pi$ than the two of Figure \ref{trian}, there is as in Figure \ref{shortest}b)
a shorter active path.
To complete the proof, we show that for the two types of  triangles shown in Figure \ref{trian}a) and
Figure \ref{trian}b)  path $\pi$  is also in \Gregtwo an active $ik$-path with respect to $c$.

\begin{figure}[H]
 \centering
           {\includegraphics[scale=.47]{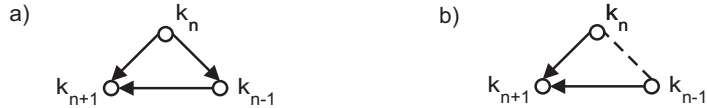}}
             \caption[]{\small The two types of triangles in regression graphs without a shorter active path whenever the path with inner nodes  $(k_{n+1}, k_n, k_{n-1}) $  is active.}
\label{trian} \end{figure}

In \Gregone containing the triangle of Figure \ref{trian}a) on a shortest  active path $\pi$,  node $k_n$ is   a transmitting node, which  is  by Lemma \ref{lem:22} outside  $c$.
By   Lemma \ref{lem:21},  node  $k_{n-1}$ is a collision node  inside $c$.
 If  instead $k_{n-1}$ were a transmitting node on $\pi$ in
 \Gregone, it would also be a transmitting node on $(k_{n-2},k_{n-1},k_{n+1})$ and give a shorter active path via the $k_{n-1}k_{n+1}$-edge, contradicting the assumption of $\pi$ being a shortest path.  Similarly, if collision node $k_{n-1}$ on $\pi$ were only an ancestor of  $c$, then there were a shorter active path via the $k_{n-1}k_{n+1}$-edge.

In addition,  node pair  $k_{n}, k_{n-2}$ is  uncoupled in \Gregone since by inserting any such edge, that is permissible in a regression graph, another shortest  path via the $k_{n-2}k_{n}$-edge would result. Therefore,
since  \Gregone and \Gregtwo have the same collision {\sf V}s, the subpath $(k_{n-2},k_{n-1},k_{n})$ forms also a collision $\sf V$ in \Gregtwo. Similarly, $(k_{n-2},k_{n-1},k_{n+1})$ is a transmitting  {\sf V} and $(k_{n+2},k_{n+1},k_{n})$ is a $\sf V$ of either type. Hence $k_{n-1}$ is a parent of $k_{n+1}$ in \Gregtwo and the only permissible edge between $k_n$ and $k_{n+1}$ is an arrow pointing to $k_{n+1}$. Therefore, $\pi$ forms an active path also in \Gregtwo.

The proof for Figure \ref{trian}b) is the same as for Figure \ref{trian}a) since the type of nodes along $\pi$, i.e.\ as collision  or transmitting nodes,  are unchanged.
\end{proof}
In the example of Figure \ref{fig:4}, all three regression graphs have the same skeleton.

\begin{figure}[H] \centering
           {\includegraphics[scale=.47]{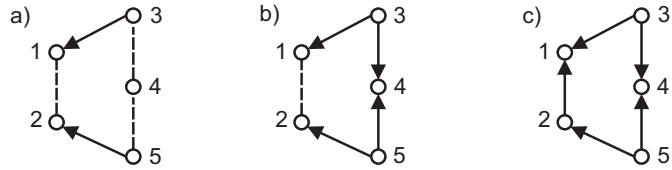}}
            \caption[]{\small{a) Regression graph $G^N_{\mathrm{reg1}}\,$, b) a Markov equivalent regression graph $G^N_{\mathrm{reg2}}\,$ to $G^N_{\mathrm{reg1}}\,$,
  c) a regression graph  $G^N_{\mathrm{reg3}}\,$ that is directed acyclic  and not Markov equivalent to $G^N_{\mathrm{reg1}}\,$.}}
        \label{fig:4}
\end{figure}

 In $G^N_{\mathrm{reg1}}\,$ there are three collision {\sf V}s $(3,4,5)$, $(1,2,5)$, and $(2,1,3)$. In $G^N_{\mathrm{reg2}}\,$ there are the same collision {\sf V}s. Therefore, these two graphs are Markov equivalent. However, there are only two  collision {\sf V}s in $G^N_{\mathrm{reg3}}\, $these are $(3,4,5)$,  and $(2,1,3)$. Hence this graph is not Markov equivalent to $G^N_{\mathrm{reg1}}\,$ and $G^N_{\mathrm{reg2}}\,$. The Markov equivalence of the graphs in Figure \ref{MEQch} to the subgraph induced by $\{b,c\}$ in
Figure \ref{childadv} are further applications of Theorem \ref{thm:1}.
Notice that Propositions 3 to 8 of Section 6 result as special cases of Theorem  \ref{thm:1}.

The following algorithm generates a directed acyclic graph from  a given \Greg that fulfills its  known necessary conditions
for Markov equivalence to a directed acyclic graph; see Proposition 2 of Wermuth (2010). We refer to these connected components  as the blocks of \Greg.\\

 %For Markov equivalence of  the subgraph induced by the nodes for context variables, this concentration graph has to be a chordal graph and the subgraph induced by the remaining nodes may not contain a  chordless collision path in four nodes.\\

\noindent{\bf Algorithm 1.} (Obtaining a Markov equivalent directed acyclic graph from a regression graph).
 \emph{Start from any given \Greg that has  a  chordal concentration graph and no
 chordless collision path in four nodes.
\begin{enumerate}
 \item Apply the maximum cardinality search algorithm on the block consisting of full lines to order the nodes of the block.
 \item Orient the edges of the block from a higher number to a lower one.
 \item Replace collision {\sf V}s by sink {\sf V}s, i.e.\ replace $i\dal\circ\dal k$ and $i\dal\circ\fla k$ by\\ $i\fra\circ\fla k$ when $i$ and $k$ are uncoupled. When a dashed line in a block  is replaced by an arrow, label the endpoints such that the arrow is from a higher number to a lower one if the labels do not already exist.
 \item Replace dashed lines  $i\dal\circ\dal k$ of  triangles  by a sink path $i\fra\circ\fla k$. When a dashed line in a block  is replaced by an arrow, label the endpoints such that the arrow is from a higher number to a lower one if the labels do not already exist.
 \item Replace dashed lines by arrows from a higher number to a lower one.
\end{enumerate}
Continually apply each step until it is not possible to continue applying it further.  Then move to the next step.}

\begin{lemma}\label{leml}
For  a regression graph  with a  chordal concentration graph and  without chordless collision paths  in four nodes,  Algorithm 1 generates a  directed acyclic graph that is Markov equivalent to \Greg.
\end{lemma}
\begin{proof}
The generated graph is directed since by  Algorithm 1, all edges are turned into arrows. Since the block containing full lines is  chordal, the graph generated by the perfect elimination order of the maximal cardinality search does not have a directed cycle; see \cite{Bla93} Section 2.4 and \cite{TarYan84}.

In addition, the arrows present in  the graph do not change by the algorithm. Thus,  to generate a cycle
containing an arrow of  the original graph, there should have been a cycle in the directed graph  generated
by replacing blocks by nodes.  But,  this is impossible  in a regression graph.  Therefore in the
generated graph,  there is no cycle containing arrows that have been between the blocks of the original graph.

Within a block, all arrows point from nodes with  higher numbers to nodes with lower ones. Otherwise,  there would have been  at step 3 of the algorithm a chordless collision path with four  nodes in the graph. Hence no directed cycle can be generated.

Theorem \ref{thm:1} gives Markov equivalence to \Greg  since Algorithm 1 preserves the skeleton of \Greg and no additional  collision {\sf V} is generated because
 sink oriented {\sf V}s remain, only dashed lines are turned into arrows and no arrows are changed to dashed lines.
\end{proof}

Notice that this algorithm does not generate a unique directed acyclic graph, but every generated directed acyclic graph is Markov equivalent to the given regression graph.
To obtain the overall complexity of Algorithm 1, we denote by $n$ the number of nodes in the graph and by  $e$ the number of edges in the graph.

\begin{coro}  The overall complexity of Algorithm 1 is $O(e^3)$. \end{coro}
\begin{proof}
Suppose that the input of Algorithm 1 is a sequence of triples, each of which consists of the two endpoints of an edge and  of the type of  edge. The length of this sequence is equal to $e$ and the highest number appearing  in the sequence is $n$. For example, the  sequence to the graph of Figure \ref{fig:4}a) is $((1,2,d),(3,1,a),(5,2,a),(4,3,d),(4,5,d))$, where `d' corresponds to a dashed line and   `a' corresponds to an arrow pointing from the first entry to the second one. Notice that this labeling is in general not  the same as the ordering of  nodes  given by  Algorithm 1.

The first two steps of Algorithm  1 can be performed in $O(e+n)$ time; see \cite{Bla93}. Step 3 of  Algorithm 1 may be performed in $e(e+1)(e-2)/2$ steps since for each edge, one can go through the edge set  to find the edges that give  a three node path with an  inner collision node.  This needs $e(e+1)/2$ steps. For each collision node, one goes again through the edge set, excluding the two edges involved in the collision path,  to check if the collision is a {\sf V}. Other actions can be done in constant time.

Step 4 may require $ne(e+1)/2$ steps since  paths considered $\circ \dal\circ\dal \circ $ which do   not form a {\sf V}. Therefore, there is no reason to go through the edge set for the third time, but one might need to go through the node ordering to decide on the direction of the generated arrow. The last step may be  performed with $ne$ steps by going through the edge set changing 'd's to 'a's appropriately by looking at the node ordering. Therefore, the overall complexity of  Algorithm 1 is $O(e^3)$.
\end{proof}

Corollary 2 and Propositions 4 to 8 can now be derived as special cases  of Theorem 1 and Lemma 4. In addition by using Lemma 1, Lemma 2 and  pairwise independences, subclasses of regression graphs can be identified,  which intersect with directed acyclic graphs, with other types of chain graphs, with concentration  graphs or with covariance graphs.

\setcounter{repeat}{1}
\begin{repeatthm}\label{MEQtoDAG}
A regression graph with a chordal graph  for the context  variables  can be oriented to be   Markov equivalent to a directed acyclic graph in the same skeleton,  if and only if it does not contain any  chordless collision path in four nodes.

\end{repeatthm}

\begin{proof} Every chordal concentration graph can be oriented to be equivalent to a directed acyclic graph;
see \cite{TarYan84}.
A missing edge for node pair $i<k$ in a directed acyclic  graph means $i\ci k|>i\setminus k$, which would contradict \ref{pairw}$(iii)$ if the graph contained a semi-directed chordless collision path in four nodes. No
undirected chordless collision path in four nodes can be  fully oriented without changing a collision $\sf V$ into a transmitting $\sf V$, but \Greg can be oriented using Algorithm 1 if it contains no such path.
\end{proof}

Notice that for joint Gaussian distributions, Theorem  2  excludes Zellner's seemingly unrelated regressions
and it excludes covariance graphs that cannot be made Markov equivalent to fully directed acyclic graphs; see Proposition
\ref{cono}.

\begin{prop}\label{MEQtoAMP}
A multivariate regression graph  with connected  components $g_1, \dots g_J$ is an AMP chain graph  in the same connected components  if  and only if  the covariance graph of every connected component of responses  is complete.
\end{prop}
\begin{proof}
The  conditional relations of the  joint response nodes in an AMP chain graph coincide with those of the  regression graph with the same connected components. Furthermore,
the subgraph induced by each connected component $g_j$ of an AMP chain graph is a concentration graph given $g_{>j}$ while  in  \Greg it is a covariance graph given $g_{>j}$. By Proposition  \ref{conc}, these have to be complete for Markov equivalence.
\end{proof}

\begin{prop}\label{MEQtoLWF}
A multivariate regression graph  with connected  components $g_1, \dots g_J$ is a LWF chain graph in the same connected components   if   and only if it contains no  semi-directed chordless  collision  path in four nodes and the covariance graph of every connected  component of responses   is complete.
\end{prop}
\begin{proof}
The proof for the connected components of a LWF chain graph is the same as for an AMP chain graph since they both have
concentration graphs for $g_j$ given $g_{>j}$.
 The  dependences of  joint responses
$g_j$ on $g_{>j}$
coincide in a LWF chain graph with the bipartite part of  the concentration graph  in $g_j \cup g_{>j}$ so that Markov equivalent independence statements can only hold with these bipartite  graphs being complete.\end{proof}

Figure \ref{intersect} illustrates Propositions  \ref{MEQtoDAG}  to \ref{MEQtoLWF} with  modified graphs of Figure \ref{hypfigaft}.  \begin{figure}[H] \centering
           {\includegraphics[scale=.47]{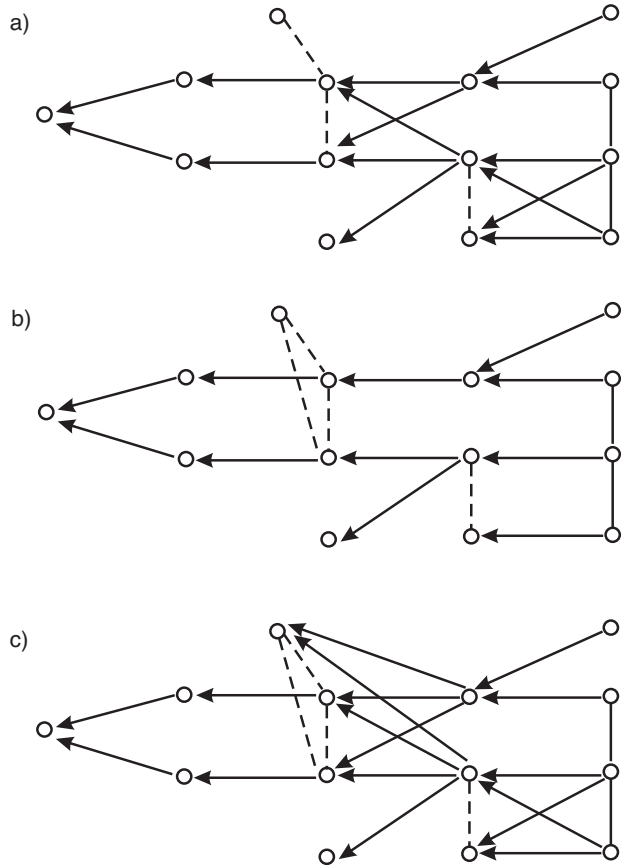}}
            \caption[]{\small{The graph of Figure \ref{hypfigaft} modified by adding edges to obtain a graph that is Markov equivalent  to a) a directed acyclic graph b) an  AMP chain graph in the same connected components c)  a LWF chain graph in the same connected components. }}
        \label{intersect}
\end{figure}

The graphs in Figure  \ref{intersect} are Markov equivalent to a) a directed acyclic graph with the same skeleton obtainable  by Algorithm 1, b) an AMP chain graph in the same connected components and  c) a LWF chain graph in the same connected components.

In general, by inserting some edges, a regression graph model  can be turned into a model in one of the  intersecting classes used  in Propositions \ref{MEQtoDAG}  to \ref{MEQtoLWF}, just as a non-chordal graph may be turned into chordal one by adding edges. When the independence structure of interest is captured by  an edge-minimal
regression graph,  then the resulting graph after adding edges will no longer be an edge-minimal graph and hence will not give the most compact graphical description possible.

However, the graph  with some added edges may define
a covering model that is easier to fit than the reduced model corresponding to the edge-minimal graph, just as an unconstrained Gaussian bivariate response regression on two regressors may be fitted in closed form, while the maximum-likelihood fitting  in the reduced model of Zellner's seemingly unrelated regression  requires iterative fitting algorithms. Any well-fitting covering model in the three  intersecting classes  will show week dependences for the edges that are to be  removed to obtain an edge-minimal graph.

Notice that sequences of regressions in the intersecting class with LWF chain graphs  correspond for Gaussian distributions
to sequences of the general linear models of \cite{And58}, Chapter 8, that is to models in which each joint response  has the same set of
regressor variables. This shows in   \Greg by  identical sets of nodes from which  arrows point to each node
within a connected component.

In contrast, the models in the intersecting classes with the two types of undirected graph may be quite complex in the sense of including many merely generated chordless cycles of size four or larger.

\begin{prop}\label{MEQtoCON}
A multivariate regression graph  has the skeleton concentration graph  if  and only if it contains  no  collision $\sf V$ and it has
the skeleton covariance graph if and only if it contains no transmitting $\sf V$.
\end{prop}

\begin{proof}
Every $\sf V$  is a collision $\sf V$  in a covariance graph  and a transmitting $\sf V$ in a concentration graph; see Lemma 1 and Lemma 2. The first includes, the second excludes the inner node from the defining independence
statement. Thus, in the presence of a {\sf V}, one would contradict the uniqueness of the defining pairwise independences.
\end{proof}

Lastly,  Figure \ref{indconc} shows the overall concentration graph induced by \Greg of Figure \ref{hypfigaft}.
It may  be obtained from
the given \Greg by finding first the smallest  covering LWF chain graph in the same connected components, then closing every sink {\sf V} by an edge, i.e.\ adding an edge between its endpoints, and finally
changing all edges to full lines.

In such a graph,  several chordless cycles in four or more nodes may be induced and the connected components of \Greg
may no longer show. In such a case, much  of the important structure of  the generating regression graph is lost.   In addition, merely induced chordless cycles  require iterative algorithms for maximum-likelihood estimation, even for  Gaussian distributions. Thus, in the case of connected  joint responses,  it may be unwise  to use a model search within the class of concentration graph models.

\begin{figure}[H] \centering
           {\includegraphics[scale=.47]{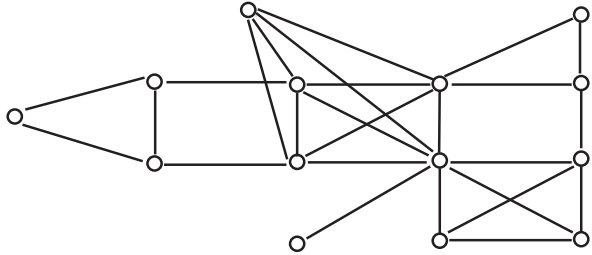}}
            \caption[]{\small{The overall concentration graph induced by the regression graph in Figure \ref{hypfigaft}.}}        \label{indconc}
\end{figure}

This  contrasts with LWF chain graphs that coincide with  regression graphs, such as in Figure \ref{intersect}c).  These preserve the available prior knowledge about the connected components and give Markov equivalence to  directed acyclic graphs so that model fitting  is possible in terms of single response regressions, that is by using just univariate conditional densities.  In addition, the simplified  criteria for Markov equivalence of  directed acyclic graphs apply.

On the other hand, sequences of regressions that coincide with  LWF chains, permit us to model simultaneous intervention on a set of variables since the corresponding independence graphs are directed and acyclic  in   nodes representing vector variables. This represents a conceptually much needed
extension of distributions generated over directed acyclic graphs in  nodes representing single variables, but excludes the more specialized seemingly unrelated regressions and incomplete covariance graphs. \\

\noindent {\bf \large Appendix: Details of regressions for the chronic pain data\\}

 The following tables show the results of  linear least-squares regressions or logistic regressions, one at a time, for each of the response variables and for each component of a joint response separately. At first, each response is regressed on all its potentially explanatory variables given by their  first ordering. The tables give the estimated constant term and for each variable in the regression, its estimated coefficient
(coeff), the estimated  standard deviation of the coefficient ($s_{\rm coeff})$, as well as the ratio $z_{\rm obs}=$coeff/$s_{\rm coeff}$. These ratios are  compared with 2.58, the  0.995 quantile of a random variable $Z$ having a standard Gaussian distribution,  for which $\Pr(|Z|>2.58)=0.01$. In backward selection steps, the variable
with the smallest observed value $|z_{obs}|$ is deleted from a regression equation, one at a time, until the threshold is reached.

\begin{table}[H]
\centering
\vspace{2mm}
%\begin{center}
\small
\setlength{\tabcolsep}{1mm}
\begin{tabular}{l P{-2,2} P{1,2} P{-1,2} c P{-2,2} P{1,2} P{-1,2} c P{-1,2}}
\toprule
\multicolumn{10}{l}{Response: $Y$, success of treatment; linear regression including a quadratic term}\\
\midrule
& \multicolumn{3}{c}{starting model} && \multicolumn{3}{c}{selected} && \ccolhd{excluded}\\
\cmidrule{2-4} \cmidrule{6-8}
explanatory variables & \ccolhd{coeff} & \ccolhd{$s_{\rm coeff}$} & \ccolhd{$z_{\rm obs}$} &&
\ccolhd{coeff} & \ccolhd{$s_{\rm coeff}$} & \ccolhd{$z_{\rm obs}$} && \ccolhd{$z'_{\rm obs}$}\\
\midrule
constant & 23.40  & -\nn&-\nn&& 20.50 &-\nn & -\nn &&-\nn\\
$Z_a$, pain intensity after&  -1.73& 0.15 & -11.19 && -1.89 & 0.15 & -12.77 && -\nn\\
$X_a$, depression after &  -0.16 & 0.05 &  -3.04&& - \nn &  -\nn&  -\nn && -1.86\\
$Z_b$, pain intensity before&  0.04 & 0.16 & 0.26 && - \nn &  -\nn&  -\nn && 0.65\\
$X_b$, depression before &   0.10 & 0.05 &  1.82 && - \nn &  -\nn&  -\nn &&  0.33\\
$U$, pain chronicity& -0.15 & 0.30 & -0.51  && - \nn & - \nn &  -\nn && -0.99\\
$A$, site of pain& -2.27 & 0.91 & -2.48 && - \nn& - \nn &  -\nn&& -2.33\\
$V$, previous illnesses& 0.19& 0.11&1.76&& - \nn& - \nn &  -\nn&&1.24\\
$B$, level of schooling&  -0.50 & 0.78 & -0.64&& - \nn &   -\nn& -\nn && -0.22\\[1mm]
\hdashline\\[-3mm]
$(Z_a-{\rm mean}(Z_a))^2$& 0.18 & 0.23 & 3.41 && 0.23 & 0.05 & 4.28 && -\nn\\
\midrule
\multicolumn{10}{l}{$R^2_{\rm full}=0.54$\n\nn Selected model$\n Y: Z_a+Z_a^2$ \n \nn$R^2_{\rm sel}=0.49$}\\
\bottomrule \\[-3mm]
\end{tabular}
%\end{center}
\label{respY}
\end{table}

\begin{table}[H]
\centering
\vspace{2mm}
%\begin{center}
\small
\setlength{\tabcolsep}{1mm}
\begin{tabular}{l P{-2,2} P{1,2} P{-1,2} c P{-2,2} P{1,2} P{-1,2} c P{-1,2}}
\toprule
\multicolumn{10}{l}{Response: $Z_a$,  intensity of pain after treatment; linear regression}\\
\midrule
& \multicolumn{3}{c}{starting model} && \multicolumn{3}{c}{selected} && \ccolhd{excluded}\\
\cmidrule{2-4} \cmidrule{6-8}
explanatory variables & \ccolhd{coeff} & \ccolhd{$s_{\rm coeff}$} & \ccolhd{$z_{\rm obs}$} &&
\ccolhd{coeff} & \ccolhd{$s_{\rm coeff}$} & \ccolhd{$z_{\rm obs}$} && \ccolhd{$z'_{\rm obs}$}\\
\midrule
constant & 2.74  & -\nn&-\nn&& 2.98 &-\nn & -\nn &&-\nn\\
$Z_b$, pain intensity before&  0.12 & 0.08 & 1.60 &&  0.16 &  0.07& 2.16$*$ && -\nn\\
$X_b$, depression before &   0.03 & 0.02 &  1.28 && - \nn &  -\nn&  -\nn &&  1.76\\
$U$, pain chronicity& 0.11& 0.14 & 0.75  && - \nn & - \nn &  -\nn && 1.43\\
$A$, site of pain& 1.07 & 0.42 & 2.51 &&  1.27 &  0.39& 3.26\n&& -\nn\\
$V$, previous illnesses& 0.00& 0.05&0.03&& - \nn& - \nn &  -\nn&&0.83\\
$B$, level of schooling&  -0.19 & 0.37 & -0.52&& - \nn &   -\nn& -\nn &&-0.70 \\[1mm]
\midrule
\multicolumn{10}{l}{$R^2_{\rm full}=0.09$ \n  \nn Selected model$\n Z_a: Z_b+A$ \n  \nn$R^2_{\rm sel}=0.07$}\\
%\midrule
\multicolumn{10}{l}{$*$: depression before treatment needed because of  the repeated measurement design;}\\
\multicolumn{10}{l}{the low correlation for $Z_a, Z_b$ is due to a change in measuring, before and after treatment}\\
\bottomrule \\[-3mm]
\end{tabular}
\label{respZa}
\end{table}

The procedure defines a  selected model, unless one of the excluded variables has a contribution of
$|z^{'}_{\rm obs}|>2.58$ when added alone to the selected directly explanatory variables, then such a variable
needs also to be included as an important directly explanatory variable. This did not happen in the given data set.

 The tables show for  linear models also $R^2$, the coefficient of determination, both for the full and for the selected model. Multiplied by 100, it gives the percentage of the variation in the response explained by the model.

\begin{table}[H]
\centering
\vspace{2mm}
\small
\setlength{\tabcolsep}{1mm}
\begin{tabular}{l P{-2,2} P{1,2} P{-1,2} c P{-2,2} P{1,2} P{-1,2} c P{-1,2}}
\toprule
\multicolumn{10}{l}{Response: $X_a$,  depression after treatment; linear regression}\\
\midrule
& \multicolumn{3}{c}{starting model} && \multicolumn{3}{c}{selected} && \ccolhd{excluded}\\
\cmidrule{2-4} \cmidrule{6-8}
explanatory variables & \ccolhd{coeff} & \ccolhd{$s_{\rm coeff}$} & \ccolhd{$z_{\rm obs}$} &&
\ccolhd{coeff} & \ccolhd{$s_{\rm coeff}$} & \ccolhd{$z_{\rm obs}$} && \ccolhd{$z'_{\rm obs}$}\\
\midrule
constant & 2.54  & -\nn&-\nn&& 4.55 &-\nn & -\nn &&-\nn\\
$Z_b$, pain intensity before& -0.05& 0.22 & -0.23 && -\nn &  -\nn& -\nn && -0.21\\
$X_b$, depression before &  0.62 & 0.06 & 10.43 && 0.68 &  0.05&  12.68 && -\nn\\
$U$, pain chronicity& 0.96& 0.42 & 2.28  && - \nn & - \nn &  -\nn &&2.31 \\
$A$, site of pain& -1.19 & 1.25 & -0.95 &&  -\nn & -\nn& -\nn&& -0.10\\
$V$, previous illnesses& 0.05& 0.15&0.35&& - \nn& - \nn &  -\nn&&1.08\\
$B$, level of schooling&  0.15 & 1.09 & 0.14&& - \nn &   -\nn& -\nn && -0.01\\[1mm]
\midrule
\multicolumn{10}{l}{$R^2_{\rm full}=0.46$ \n  \nn Selected model$\n X_a: X_b$ \n  \nn$R^2_{\rm sel}=0.45$}\\
\bottomrule \\[-3mm]
\end{tabular}
\label{respXa}
\end{table}

\begin{table}[H]
\centering
\vspace{2mm}
\small
\setlength{\tabcolsep}{1mm}
\begin{tabular}{l P{-2,2} P{1,2} P{-1,2} c P{-2,2} P{1,2} P{-1,2} c P{-1,2}}
\toprule
\multicolumn{10}{l}{Response: $Z_b$, intensity of pain before; linear regression}\\
\midrule
& \multicolumn{3}{c}{starting model} && \multicolumn{3}{c}{selected} && \ccolhd{excluded}\\
\cmidrule{2-4} \cmidrule{6-8}
explanatory variables & \ccolhd{coeff} & \ccolhd{$s_{\rm coeff}$} & \ccolhd{$z_{\rm obs}$} &&
\ccolhd{coeff} & \ccolhd{$s_{\rm coeff}$} & \ccolhd{$z_{\rm obs}$} && \ccolhd{$z'_{\rm obs}$}\\
\midrule
constant & 7.60  & -\nn&-\nn&& 7.38 &-\nn & -\nn &&-\nn\\
$U$, pain chronicity& 0.10& 0.13 & 0.77  && - \nn & - \nn &  -\nn &&0.59 \\
$A$, site of pain& -0.58 & 0.40 & -1.44 && -\nn &  -\nn& -\nn&& -1.20\\
$V$, previous illnesses& 0.02& 0.05&0.46&& - \nn& - \nn &  -\nn&&0.72\\
$B$, level of schooling&  -0.94 & 0.35 & -2.70&& -0.89 &   0.33& -2.65 && -\nn\\[1mm]
\midrule
\multicolumn{10}{l}{$R^2_{\rm full}=0.05$ \n  \nn Selected model$\n Z_a:  B$ \n  \nn$R^2_{\rm sel}=0.03$}\\
\bottomrule \\[-3mm]
\end{tabular}
\label{respZb}
\end{table}

In the linear regression of $Z_a$ on  $X_a$ and on  the directly explanatory variables of both $Z_a$ and $X_a$,
that is on $Z_b,   X_b,A$,  the contribution of $X_a$  leads to $z_{\rm obs}=3.51$, which  coincides -- by definition --  with  $z_{\rm obs}$ computed for the  contribution of $Z_a$ in the linear regression of $X_a$ on $Z_a$ and on $Z_b,    X_b,A$.   Hence the two responses are correlated
even after considering the directly explanatory variables and a dashed line joining  $Z_a$ and $Z_b$ is added to  the well-fitting regression graph in Figure \ref{figregpain}.

In the linear regression of $Z_b$ on  $X_b$ and on the directly explanatory variables of both $Z_b$ and $X_b$,
that is on  $U,  A,  V, B$,  the contribution of $X_b$  leads to $z_{\rm obs}=2.64$.
  Hence the two responses are associated after considering their directly explanatory variables and there is a dashed line joining $Z_b$ and $X_b$ in the regression graph of Figure \ref{figregpain}.\\[-7mm]

The relatively strict criterion, for excluding variables, assures that all edges in the derived regression graph
correspond to dependences and dependences that are considered to be substantive in the given context.
Had instead a 0.975 quantile been chosen as  threshold, then one arrow from $A$ to $Y$ and another
from $U$ to $X_a$ would have been added to the regression graph. Though this would correspond to a better  goodness-of-fit, such weak dependences are less likely to become  confirmed as being important  in follow-up studies.

\begin{table}[H]
\centering
\vspace{2mm}
\small
\setlength{\tabcolsep}{1mm}
\begin{tabular}{l P{-2,2} P{1,2} P{-1,2} c P{-2,2} P{1,2} P{-1,2} c P{-1,2}}
\toprule
\multicolumn{10}{l}{Response: $X_b$, depression  before; linear regression}\\
\midrule
& \multicolumn{3}{c}{starting model} && \multicolumn{3}{c}{selected} && \ccolhd{excluded}\\
\cmidrule{2-4} \cmidrule{6-8}
explanatory variables & \ccolhd{coeff} & \ccolhd{$s_{\rm coeff}$} & \ccolhd{$z_{\rm obs}$} &&
\ccolhd{coeff} & \ccolhd{$s_{\rm coeff}$} & \ccolhd{$z_{\rm obs}$} && \ccolhd{$z'_{\rm obs}$}\\
\midrule
constant &10.96 & -\nn&-\nn&& 7.31 &-\nn & -\nn &&-\nn\\
$U$, pain chronicity& 1.97& 0.49 & 4.02  && 1.78 & 0.46 &  3.87 &&-\nn \\
$A$, site of pain& -2.33 & 1.50 & -1.55 && -\nn &  -\nn& -\nn&& -1.42\\
$V$, previous illnesses& 0.54& 0.18&2.99&& 0.55& 0.18&  3.06&&-\nn\\
$B$, level of schooling&  -1.10 & 1.31 & -0.84&& -\nn &  -\nn & -\nn && -0.57\\[1mm]
\midrule
\multicolumn{10}{l}{$R^2_{\rm full}=0.18$ \n  \nn Selected model$\n X_b:  U+V$ \n  \nn$R^2_{\rm sel}=0.17$}\\
\bottomrule \\[-6mm]
\end{tabular}
\label{respXb}
\end{table}

\begin{table}[H]
\centering
\vspace{2mm}
\small
\setlength{\tabcolsep}{1mm}
\begin{tabular}{l P{-2,2} P{1,2} P{-1,2} c P{-2,2} P{1,2} P{-1,2} c P{-1,2}}
\toprule
\multicolumn{10}{l}{Response: $U$, chronicity of pain; linear regression}\\
\midrule
& \multicolumn{3}{c}{starting model} && \multicolumn{3}{c}{selected} && \ccolhd{excluded}\\
\cmidrule{2-4} \cmidrule{6-8}
explanatory variables & \ccolhd{coeff} & \ccolhd{$s_{\rm coeff}$} & \ccolhd{$z_{\rm obs}$} &&
\ccolhd{coeff} & \ccolhd{$s_{\rm coeff}$} & \ccolhd{$z_{\rm obs}$} && \ccolhd{$z'_{\rm obs}$}\\
\midrule
constant &2.93 & -\nn&-\nn&& 2.47&-\nn & -\nn &&-\nn\\
$A$, site of pain& 0.95 & 0.21 & 4.58 && 1.02 &  0.20& 5.02&& -\nn\\
$V$, previous illnesses& 0.14& 0.02&5.83&& 0.14& 0.02&  5.92&&-\nn\\
$B$, level of schooling&  -0.27& 0.19 & -1.43&& -\nn &  -\nn & -\nn && -1.43\\[1mm]
\midrule
\multicolumn{10}{l}{$R^2_{\rm full}=0.26$ \n  \nn Selected model$\n X_b:  A+V$ \n  \nn$R^2_{\rm sel}=0.25$}\\
\bottomrule \\[-6mm]
\end{tabular}
\label{respU}
\end{table}

\begin{table}[H]
\centering
\vspace{2mm}
\small
\setlength{\tabcolsep}{1mm}
\begin{tabular}{l P{-2,2} P{1,2} P{-1,2} c P{-2,2} P{1,2} P{-1,2} c P{-1,2}}
\toprule
\multicolumn{10}{l}{Response: $A$,  site of pain; logistic regression}\\
\midrule
& \multicolumn{3}{c}{starting model} && \multicolumn{3}{c}{selected} && \ccolhd{excluded}\\
\cmidrule{2-4} \cmidrule{6-8}
explanatory variables & \ccolhd{coeff} & \ccolhd{$s_{\rm coeff}$} & \ccolhd{$z_{\rm obs}$} &&
\ccolhd{coeff} & \ccolhd{$s_{\rm coeff}$} & \ccolhd{$z_{\rm obs}$} && \ccolhd{$z'_{\rm obs}$}\\
\midrule
constant &0.26 & -\nn&-\nn&& 0.60 &-\nn & -\nn &&-\nn\\
$V$, previous illnesses& 0.05& 0.04&1.22&& -\nn& -\nn &  -\nn &&1.22\\
$B$, level of schooling&  -1.25 & 0.40& -3.11&& -1.28&  0.40 & -3.18 && -\nn\\[1mm]
\midrule
\multicolumn{10}{l}{Selected model\n $A:  B$; response recoded to (0,1)  instead of (1,2) }\\
\bottomrule \\[-3mm]
\end{tabular}
\label{respA}
\end{table}

\begin{table}[H]
\centering
\vspace{2mm}
\small
\setlength{\tabcolsep}{1mm}
\begin{tabular}{l P{-2,2} P{1,2} P{-1,2} c P{-2,2} P{1,2} P{-1,2} c P{-1,2}}
\toprule
\multicolumn{10}{l}{Response: $V$,  previous illnesses; linear regression}\\
\midrule
& \multicolumn{3}{c}{starting model} && \multicolumn{3}{c}{selected} && \ccolhd{excluded}\\
\cmidrule{2-4} \cmidrule{6-8}
explanatory variables & \ccolhd{coeff} & \ccolhd{$s_{\rm coeff}$} & \ccolhd{$z_{\rm obs}$} &&
\ccolhd{coeff} & \ccolhd{$s_{\rm coeff}$} & \ccolhd{$z_{\rm obs}$} && \ccolhd{$z'_{\rm obs}$}\\
\midrule
constant &6.41& -\nn&-\nn&& 5.53 &-\nn & -\nn &&-\nn\\
$B$, level of schooling&  -0.65 & 0.54 & -1.20&& -\nn &  -\nn & -\nn && -\nn\\[1mm]
\midrule
\multicolumn{10}{l}{Selected model\n $V:  -$ }\\
\bottomrule \\[-3mm]
\end{tabular}
\label{respV}
\end{table}

 The subgraph induced by $Z_a, Z_b, X_a, X_b$ of the regression graph in Figure \ref{figregpain} corresponds to two seemingly unrelated regressions. Even though separate least-squares estimates can in principle be severely distorted,
 for the present data, the structure is   so well-fitting in the unconstrained multivariate regression of   $Z_a$ and $X_a$ on $Z_b$, $X_b$, $U,V,A,B$, that is in a simple covering model, that none of these potential problems are relevant.

  With $C=\{U,V,A,B\}$, this is evident from the observed covariance matrix of $Z_a, X_a$ given
  $Z_b, X_b, C$, denoted here by $\tilde{\Sigma}_{aa|bC}$ and the observed regression coefficient matrix
  $\tilde{\Pi}_{a|b.C}$ being almost identical  to the corresponding m.l.e  $\hat{\Sigma}_{aa|bC}$ and $\hat{\Pi}_{a|b.C}$.

  The former can be obtained by sweeping or partially inverting  the observed covariance matrix of the eight variables with respect to $Z_b, X_b, C$ and the latter by using an adaption of the EM-algorithm,
  due to Kiiveri (1989),  on the observed covariance matrix of the four symptoms,  corrected for linear regression on $C$. In this way, one gets
  $$ \tilde{\Sigma}_{aa|bC}= \left(\begin{array}{rr}5.61& 3.91\\3.91& 48.37 \end{array} \right), \nn \nn
  \hat{\Sigma}_{aa|bC}= \left(\begin{array}{rr}5.66& 3.94\\3.94& 48.41 \end{array} \right), $$

   $$ \tilde{\Pi}_{a|b.C}= \left(\begin{array}{rr}0.12 & 0.03\\-0.05& 0.62 \end{array} \right), \nn \nn
  \hat{\Pi}_{a|bC}= \left(\begin{array}{rr}0.14& 0.00\\0.00 & 0.60 \end{array} \right).$$

  The assumed definition of  the joint distribution in terms of  univariate and multivariate regressions assures that  the overall fit of the model can be judged locally in two steps. First, one   compares each unconstrained, full regression
  of a single response with regressions constrained by some independences, that is by
  selecting  a  subset of directly  explanatory variables from the list of the potentially explanatory variables.
  Next, one decides for each component pair of a joint response whether this pair is  conditionally  independent given their  directly explanatory variables considered jointly. This can again be achieved by single univariate regressions, as illustrated above  for the joint responses $Z_a$ and $X_a$.  \\

\noindent{\bf Acknowledgement.} The work of the first author has been supported  in part by the Swedish Research Society via the Gothenburg Stochastic Center and by the Swedish Strategic  Fund via  the Gothenburg Mathematical  Modeling Center.  We thank R. Castelo, D.R. Cox,  G. Marchetti and the referees for their most helpful comments.\\

\renewcommand{\baselinestretch}{1.2}


\begin{thebibliography}{100}
{\small
\bibitem[Ali, Richardson and Spirtes(2009)]{AliRicSpi09}
{Ali, R. A., Richardson, T. S. and Spirtes, P.} (2009).
 Markov equivalence for ancestral graphs.
\textit{Ann. Statist.}
  \textbf{37}, 2808--2837.
\vspace{-2mm}



\bibitem[Anderson(1958)] {And58}
 {Anderson, T. W.}  (1958).
 \textit{An introduction to
multivariate statistical analysis.} (3rd ed., 2003) Wiley, New York.
\vspace{-2mm}

\bibitem[Anderson(1973)] {And73}
 Anderson, T. W. (1973). Asymptotically efficient
estimation of covariance matrices with linear structure.
{\em Ann.\ Statist.} {\bf 1}, 135-141.
\vspace{-2mm}



\bibitem[Andersson and Perlman(2006)]{AndPer06}
{Andersson, S. A. and Perlman, M. D.}  (2006).
Characterizing Markov equivalence classes for AMP chain graph.
\textit{Ann. Statist.}  \textbf{34}, 939--972.
\vspace{-2mm}

\bibitem[Andersson, Madigan and Perlman(2001)]{AndMadPer97}
{Andersson, S. A.,  Madigan, D., and Perlman, M. D.}
 (2001).  Alternative Markov properties  for chain graphs.
{\em Scand. J. Statist.} {\bf 28},  33--86.
\vspace{-2mm}

\bibitem[Andersson, Madigan, Perlman and Triggs(1997)]{AndMadPerTr97}
{Andersson, S. A.,  Madigan D., Perlman M. D.  and
Triggs, C.M.} (1997).  A graphical characterization of lattice conditional independence models.
\textit{Ann. Math.  Artif. Intell.}
\textbf{21}, 27--50.
\vspace{-8mm}

\bibitem[Barndorff-Nielsen(1978)]{Barn78}
{Barndorff-Nielsen, O. E.}  (1978).
\textit{Information and exponential families in statistical theory}. Wiley,  Chichester.
\vspace{-2mm}



\bibitem[Bergsma and Rudas(2002)]{BerRud02}
{Bergsma, W. and Rudas, T.} (2002).
Marginal models for categorical data.
\textit{Ann. Statist.}
\textbf{30}, 140--159.
\vspace{-2mm}


\bibitem[Birch(1963)]{Birch63}
{Birch, M. W.} (1963). Maximum likelihood in three-way contingency tables.
\textit{J. Roy. Statist. Soc. B}
\textbf{25}, 220--233.
\vspace{-2mm}


\bibitem[Bishop,  Fienberg and Holland(1975)]{BisFieHol75}
{Bishop, Y. M. M., Fienberg, S. F. and Holland, P. W.} (1975).
\textit{Discrete multivariate analysis}.
 MIT Press,  Cambridge.
\vspace{-2mm}

 \bibitem[Blair and Peyton(1993)]{Bla93}
{Blair,  J. R. S.  and Peyton,  B. W. }(1993).
An introduction to chordal graphs and clique trees. In: \textit{Graph theory and sparse matrix computations},
IMA Volumes in mathematics and its applications, Vol. 56. (eds.J. A. George, J. R. Gilbert, and J. W. H. Liu)
Springer, New York, 1--30.
\vspace{-2mm}



\bibitem[Brito and Pearl(2002)]{BriPea02}
{Brito, C. and  Pearl, J.} (2002).
A new identification condition for recursive models with correlated errors.
\textit{Structural Equation Model.}
\textbf{9}, 459--474.
\vspace{-2mm}

 \bibitem[Bollen(1989)]{Boll89}
 {Bollen, K. A.} (1989).
 \textit{Structural equations with latent variables.}
Wiley, New York.
\vspace{-2mm}





\bibitem[Brown(1986)]{Brown86}
{Brown, L. D.} (1986).
 \textit{Fundamentals of statistical exponential families with
applications in statistical decision theory.}
LNMS {\bf  9},   Inst. Math. Statist., Beachwood.
\vspace{-2mm}


\bibitem[Castelo and Kocka(2003)]{CaKo03}
{Castelo, R. and  Kocka, T.} (2003).
On inclusion-driven learning of Bayesian networks.
\textit{J. Machine Learn. Res.}
\textbf{4}, 527--574.
\vspace{-2mm}



\bibitem[Castelo and Siebes(2003)]{CasSie03}
{Castelo, R. and  Siebes. A.} (2003).
A characterization of moral transitive acyclic directed graph Markov models as labeled trees.
\textit{J. Stat. Plan. Inf.}
\textbf{115}, 235--259.
\vspace{-2mm}




\bibitem[Caussinus(1966)]{Causs66}
{Caussinus, H.}(1966).
Contribution \'a l'analyse statistique des tableaux de corr\'elation.
   \textit{Annales de la Facult\'e des Sciences de l'Universit\'e de Toulouse}
    \textbf{29},  77--183.
\vspace{-2mm}



\bibitem[Cayley(1889)]{Cay889}
{Cayley, A.} (1889).
A theorem on trees.
\textit{Quart. J. Math.}
\textbf{ 23}, 376--378.
\vspace{-2mm}



\bibitem[Chaudhuri, Drton and Richardson(2007)]{ChaDrtRich07}
{Chaudhuri, S.,  Drton, M. and Richardson, T. S.} (2007).
Estimation of a covariance matrix with zeros
\textit{Biometrika}  \textbf{94}, 199--216.
\vspace{-2mm}

\bibitem[Cochran(1938)]{Coch38}
{Cochran, W. G.}  (1938). The omission or addition of an independent
variate in multiple linear regression.
\textit{Suppl. J. Roy. Statist. Soc.} {\bf 5}, 171--176.
\vspace{-2mm}

\bibitem[Chickering(1995)]{Chi95}
{Chickering, D. M.}  (1995).
A transformational characterization of equivalent Bayesian networks.
In \textit{Proc. 10th UAI conf.} (eds. P. Besnard and S. Hanks)
Morgan  Kaufman, San Mateo, 87--98.
\vspace{-2mm}

\bibitem[Cox(1966)]{Cox66}
{Cox, D. R.}  (1966).
Some procedures associated with the logistic qualitative response curve. In \textit{Research papers in
statistics: Festschrift for J. Neyman} (ed. F.N.  David), Wiley, New York,   55--71.
\vspace{-2mm}

\bibitem[Cox(2006)]{Cox06}
{Cox, D. R.}(2006).
\textit{Principles of statistical inference}.
Cambridge University Press, Cambridge.
\vspace{-2mm}

\bibitem[Cox and Wermuth(1990)]{CoxWer90}
{Cox, D. R. and Wermuth, N.} (1990).
An approximation to
maximum-likelihood estimates in reduced models.
\textit{Biometrika}
\textbf {77}, 747--761.
\vspace{-2mm}



\bibitem[Cox and Wermuth(1993)]{CoxWer93}
{Cox, D. R. and Wermuth, N.} (1993).
Linear dependencies represented by
chain graphs (with discussion).
\textit{Statist. Science}
\textbf{8}, 204--218;
247--277.
\vspace{-2mm}


\bibitem[Cox and  Wermuth(1996)]{CoxWer96}
{Cox, D. R. and Wermuth, N.} (1996).
\textit{Multivariate dependencies: models, analysis, and interpretation}.
Chapman and Hall (CRC), London.
\vspace{-2mm}


\bibitem[Cox and Wermuth(1999)]{CoxWer99}
{Cox, D. R. and Wermuth, N.} (1999).
Likelihood factorizations  for mixed discrete and continuous variables.
\textit{Scand. J. Statist.}
\textbf{26}, 209-220.
\vspace{-2mm}




\bibitem[Cox and Wermuth(2003)]{CoxWer03}
{Cox, D. R. and Wermuth, N.} (2003).
 A general condition for avoiding effect reversal after marginalization.
 \textit{J. Roy.  Statist.  Soc. B} \textbf{65}, 937--941.
\vspace{-2mm}

\bibitem[Darroch(1962)]{Darr62}
{Darroch, J. N.} (1962).
Interactions in multi-factor contingency tables.
 \textit{J. Roy. Statist. Soc. B}
 \textbf{24}, 251--263.
\vspace{-2mm}

\bibitem[Darroch, Lauritzen and Speed(1980)]{DaLauSpeed80}
{Darroch, J. N., Lauritzen, S. L. and Speed, T. P.}
 (1980).  Markov fields and log-linear models for contingency tables.
\textit{Ann. Statist.} \textbf {8}, 522--539.
\vspace{-2mm}

\bibitem[Dawid(1979)]{Daw79}
{Dawid, A. P.} (1979).
Conditional independence in statistical theory (with discussion). \textit{ J. Roy. Statist. Soc.  B}
\textbf{ 41}, 1--31.
\vspace{-2mm}


\bibitem[Dempster(1969)]{Dem69}
 {Dempster,  A. P.}  (1969).
 \textit{Elements of continuous
multivariate analysis.}  Addison-Wesley, Reading, Mass.
\vspace{-2mm}


\bibitem[Dempster(1972)]{Dem72}
{Dempster, A. P.} (1972).
Covariance selection.
\textit{Biometrics} \textbf{28}, 157--175.
\vspace{-2mm}


\bibitem[Dinitz(2006)]{Din06}
{Dinitz, Y.}  (2006).
Dinitz'  algorithm:  the original version and Even's  version. In:
\textit{Essays in memory of Shimon Even.} (eds.  S. Even, O. Goldreich, A.L. Rosenberg, and A.L. Selman)
 Springer, New York,  218--240.
\vspace{-2mm}

\bibitem[Dirac(1961)]{Dir61}
{Dirac, G. A.} (1961).
On rigid circuit graphs.
\textit{Abhandl. Math. Seminar Hamburg}
\textbf{ 25}, 71--76.

\vspace{-2mm}


\bibitem[Drton(2009)]{Drton09}
{Drton, M.} (2009).
Discrete chain graph models. \textit{Bernoulli} {\bf 15}, 736--753.
\vspace{-2mm}

%\bibitem[Drton, Eichler and Richardson(2009)]{DrtEicRich09}
%{Drton, M., Eichler, M. and Richardson, T.S.} (2009).
% Computing maximum likelihood estimates in recursive linear models.
%\textit{J.  Mach. Learn. Res.} \textbf{10}, 2329--2348.
%\vspace{-3mm}

\bibitem[\mbox{Drton} and Perlman(2004)]{DrtonPer04}
{Drton, M. and Perlman, M. D.}  (2004).
Model selection for Gaussian concentration graphs.
\textit{Biometrika} \textbf{91},  591--602.
\vspace{-2mm}



\bibitem[Drton and Richardson(2004)]{DrtRich04}
{Drton,  M.  and  Richardson, T. S.}  (2004).
Multimodality of the likelihood in the bivariate seemingly unrelated regression model.
\textit{Biometrika} \textbf{91}, 383--392.
\vspace{-2mm}

\bibitem[Drton and Richardson(2008a)]{DrtRich08}
{Drton, M. and Richardson, T. S.}  (2008a). Binary models for marginal independence.
\textit{J. Roy. Statist. Soc.
B} \textbf{70},  287--309.
\vspace{-2mm}


 \bibitem[Drton and Richardson(2008b)]{DrtonRich08}
 {Drton, M. and Richardson, T. S.}  (2008b).	
 Graphical methods for efficient likelihood inference in Gaussian covariance models  J.
 \textit{J. Machine Learn. Res.}
 \textbf{9},  893--914.
\vspace{-2mm}



\bibitem[Edwards(2000)] {Edw00}
{Edwards, D.} (2000).
\textit{Introduction to graphical modelling.} (2nd ed.) Springer, New York.
\vspace{-2mm}



\bibitem[Foygel, Draisma and Drton(2011)] {FoyDraDrt11}
{Foygel, R., Draisma, J. and Drton, M.} (2011).
Half-trek criterion for generic identifiability of linear structural equation models.
\textit{Submitted and available under:}
http://arxiv.org/abs/1107.5552
\vspace{-2mm}

 \bibitem[Frydenberg(1990)]{Fryd90}
 Frydenberg, M. (1990). The chain graph Markov property.
 \textit{Scand.  J.
 Statist.} \textbf{17}, 333--353.
\vspace{-2mm}





 \bibitem[Geiger, Verma and Pearl(1990)]{GeiVerPea90}
   {Geiger, D., Verma, T. S.  and  Pearl, J.}  (1990). Identifying independence in Bayesian networks.
   \textit{Networks}   \textbf{20},  507--534.
\vspace{-2mm}


   \bibitem[Glonek and McCullagh(1995)]{GlonMcCul95}
{Glonek    G. F. V.  and McCullagh P.} (1995).
 Multivariate logistic models.
 \textit {J. Roy. Statist. Soc. B} \textbf{53}, 533--546.
\vspace{-2mm}


\bibitem[Goodman(1970)]{Goodm70}
{Goodman, L. A.} (1970).
The multivariate analysis of qualitative data: interaction among multiple classifications.
\textit{J. Amer. Statist. Assoc.}
\textbf {65}, 226--256.
\vspace{-2mm}

\bibitem[Hardt et al.(2008)]{Hardt2008}
{Hardt, J., Sidor, A., Nickel, R., Kappis, B., Petrak, F., and  Egle, U. T.}  (2008).
Childhood adversities and suicide attempts: a retrospective study.
\textit{J. Family Violence},
\textbf{23}, 713-718.
\vspace{-2mm}

\bibitem[Haavelmo(1943)] {Havelm43}
{Haavelmo, T.} (1943).
The statistical implications of a system of simultaneous equations.
\textit{Econometrica}
 \textbf{11}, 1-12.
\vspace{-2mm}


%\bibitem[Holland and Rosenbaum (1986)]{HolRos86}
%{Holland P.W. and Rosenbaum, P.R.}  (1986).
%Conditional association and unidimensionality in monotone latent variable models.
%\textit{Ann. Statist.}
%\textbf{14}, 1523--1543.
%\vspace{-3mm}





%\bibitem[Ibarra(2008)]{Iba08}
%{Ibarra, L} (2008).
%Fully dynamic algorithms for chordal graphs and split graphs.
%\textit{ACM Trans. Algor.}
%\textbf{4}, 1--20.
%\vspace{-3mm}

\bibitem[Jensen(1988)]{Jen88}
{Jensen, S. T.}  (1988)
Covariance hypotheses which are linear in both the covariance and the inverse covariance.
\textit{Ann. Statist.}
\textbf{16}, 302--322.
\vspace{-2mm}


\bibitem[
J\"oreskog(1981)]{Jor81}
{J\"oreskog, K. G.} (1981).
Analysis of covariance structures.
\textit{Scand J Statist.} \textbf {8}, 65--92.
\vspace{-2mm}

\bibitem[Kang and Tian(2009)]{KangTian09}
{Kang, C. and Tian, J.} (2009).
 Markov properties for linear causal models with correlated errors.
 \textit{J.  Mach. Learn. Res.}
{\bf 10}, 41--70.
\vspace{-2mm}

 \bibitem[Kappesser(1997)]{Kappesser97}
 {Kappesser, J.} (1997).
 \textit{Bedeutung der Lokalisation
 f\"ur die Entwicklung und Behandlung chronischer
Schmerzen.} Thesis,  Department of Psychology,
University of Mainz.
\vspace{-2mm}


\bibitem[Kauermann(1996)]{Kau96}
{Kauermann, G.} (1996).
 On a dualization of graphical Gaussian models. \textit{ Scand. J. Statist.}
\textbf{23}, 115--116.
\vspace{-2mm}

\bibitem[Kiiveri(1987)]{Kii87}
{Kiiveri, H. T.} (1987).
An incomplete data approach to the analysis of covariance structures.
\textit{Psychometrika}
\textbf{ 52}, 539--554.
\vspace{-2mm}


\bibitem[Kiiveri, Speed and Carlin(1984)]{KiiSpeCar84}
{Kiiveri, H. T., Speed, T. P.  and Carlin, J. B. } (1984).
Recursive causal models. \textit{J.  Austral. Math. Soc  A}  \textbf{36},  30--52.
\vspace{-2mm}


\bibitem[Kline(2006)]{Kline06}
{Kline, R. B.} (2006).
\textit{Principles and practice of structural equation modeling} (3rd edition).
Guilford Press, New York.
\vspace{-2mm}



\bibitem[Koster(2002)]{Kost02}
{Koster, J.} (2002).
Marginalising and conditioning in graphical models.
\textit{Bernoulli}
\textbf{8}, 817--840.
\vspace{-2mm}


\bibitem[Lauritzen(1996)]{Lau96}
{Lauritzen, S. L.}  (1996).
\textit{Graphical Models}.
Oxford University Press, Oxford.
\vspace{-2mm}

\bibitem[Lauritzen and  Wermuth(1989)]{LauWer89}
{Lauritzen, S. L. and Wermuth, N.} (1989).
Graphical models for
association between variables, some of which are qualitative and some
quantitative.
 \textit{Ann. Statist.}
 \textbf{17}, 31--57.
\vspace{-2mm}

 \bibitem[Lehmann and Scheff\'e(1955)] {LehSch55}
 {Lehmann, E. L. and Scheff\'e, H.} (1955).
Completeness, similar regions and unbiased estimation.
\textit{Sankhya}  {\bf 14}, 219--236.
\vspace{-2mm}

\bibitem [Levitz, Perlman and Madigan(2001)]{LevPerMad01}
{ Levitz, M.,  Perlman,  M. D., \&  Madigan, D.} (2001).
Separation and completeness
properties for AMP chain graph Markov models. \textit{Ann. Statist.} \textbf{29},  1751--1784.
\vspace{-2mm}


\bibitem [Ln\v{e}ni\v{c}ka and Mat\'u\v{s}(2007)]{LnenMatus07}
{Ln\v{e}ni\v{c}ka, R. and Mat\'u\v{s}, F.} (2007).
 On Gaussian conditional independence structures.
 \textit{Kybernetika} \textbf{43},
323--342.
\vspace{-2mm}

\bibitem[Lupparelli, Marchetti and Bergsma(2009)]{LupMarBer09}
{Lupparelli M.,  Marchetti, G. M.,  and Bergsma, W. P. } (2009).
Parameterization and fitting of discrete bi-directed graph models.
\textit{Scand.  J.  Statist.}
\textbf{36}, 559--576.
\vspace{-2mm}



\bibitem[Ma, Xie and Geng(2006)]{MaXieGeng06}
{Ma, Z. M., Xie, X. C. and Geng, Z.}  (2006). Collapsibility of distribution dependence.
\textit{J. Roy. Statist. Soc. B}
\textbf{ 68}, 127--133.
\vspace{-2mm}


\bibitem[Mandelbaum and  R\"uschendorf(1987)]{ManRue87}
{Mandelbaum, A. and R\"uschendorf, L.} (1987).
Complete and symmetrically complete families of distributions.
 \textit{Ann. Statist.}
 \textbf{15}, 1229--1244.

\vspace{-2mm}



\bibitem[Marchetti and Lupparelli(2011)]{MarLup11}
{Marchetti, G. M. and Lupparelli, M.}  (2011). Chain graph models of multivariate
regression type for categorical data.
\textit{Bernoulli}
\textbf{17}, 845--879.
\vspace{-2mm}



\bibitem[Marchetti and Wermuth(2009)]{MarWer09}
{Marchetti, G. M. and Wermuth, N.} (2009).
Matrix representations and independencies in
directed acyclic graphs.
\textit{Ann. Statist.}
\textbf{47}, 961--978.
\vspace{-2mm}


\bibitem[McCullagh and Nelder(1989)]{McCullNeld89}
{McCullagh, P. and Nelder, J. A.} (1989).
\textit{Generalized Linear Models} , 2nd ed.
Chapman and Hall (CRC), London.
\vspace{-2mm}


\bibitem[Nelder and Wedderburn(1972)]{NelWed72}
{Nelder, J. A. and  Wedderburn, R.} (1972).
Generalized Linear Models.
\textit{J. Roy. Statist. Soc.   A}
\textbf{135}, 37--384.
\vspace{-2mm}



\bibitem[Pearl(1988)]{Pea88}
{Pearl, J.} (1988).
\textit{Probabilistic reasoning in intelligent systems.}
Morgan Kaufmann, San Mateo.
\vspace{-8mm}

\bibitem[Pearl(2009)]{Pea09}
Pearl, J. (2009).
\textit{ Causality: models, reasoning, and inference.}  2nd ed., Cambridge University
Press, New York.
\vspace{-2mm}

\bibitem[Pearl and Paz(1987)]{PeaPaz87}
{Pearl J. and Paz, A.} (1987). Graphoids: a graph based logic for reasoning about relevancy revelations. In:  (eds. B.D. Boulay D. Hogg, and L. Steel)
 \textit{Advances in artificial intelligence  II}, North Holland, Amsterdam, 357--363.
\vspace{-2mm}

\bibitem[Pearl and  Wermuth(1994)]{PeaWer94}
{Pearl, J. and  Wermuth, N.} (1994).
When can association graphs
admit a causal interpretation? In:
\textit{Models and
data, artificial intelligence and statistics IV.}
(eds.  P. Cheeseman and W. Oldford). Springer,  New York, 205--214.
\vspace{-2mm}



\bibitem[Richardson and Spirtes(2002)]{RichSpir02}
{Richardson, T. S. and Spirtes, P.} (2002).
Ancestral Markov graphical models.
\textit{Ann. Statist.}
\textbf{30}, 962--1030.
\vspace{-2mm}


\bibitem[Roverato(2005)]{Rov05}
{Roverato, A.} (2005).
 A unified approach to the characterisation of Markov equivalence classes of directed acyclic graphs, chain graphs with no flags and chain graphs. \textit{Scand.  J.  Statist.}
 \textbf{32},  295--312.
\vspace{-2mm}


 \bibitem[Roverato and  Studen\'y(2006)]{RovStud06}
 {Roverato, A. and Studen\'y M.}  (2006).
A graphical representation of equivalence classes of AMP chain graphs.
 \textit{J. Mach. Learn. Res.}
 \textbf{7}, 1045--1078.
\vspace{-2mm}


\bibitem[Rudas, Bergsma and Nemeth(2010)]{RudBerNem10}
{Rudas, T., Bergsma W. P.  and Nemeth, R.}
 (2010).
 Marginal log-linear parameterization of conditional independence models.
 \textit{Biometrika},  to appear.
\vspace{-2mm}

\bibitem[Sadeghi(2009)]{Sadeghi09}
{Sadeghi, K.} (2009).
Representing modified independence structures.
\textit{Transfer thesis,  Oxford University.}



\bibitem[Sadeghi and Lauritzen(2011)]{SadLau11}
{Sadeghi, K. and Lauritzen, S. L. } (2011).
Markov properties of mixed loopless graphs.
\textit{Ann. Statist.} Submitted and
\vspace{-2mm}

\bibitem[San Martin, Mouchart and Rolin(2005)]{SanMMouRol05}
{San Martin E.,  Mochart M. and Rolin, J. M.} (2005).
Ignorable common information, null sets and Basu's first theorem.
\textit{Sankhya} \textbf{67},  674--698.
\vspace{-2mm}

	

\bibitem[Speed and Kiiveri(1986)]{SpeedKiv86}
{Speed, T. P. and Kiiveri, H. T.} (1986).
Gaussian  Markov distributions over finite graphs.
\textit{Ann. Statist.}
\textbf{14}, 138--150.
\vspace{-2mm}


\bibitem[Spirtes, Glymour  and Scheines(1993)]{SpiGlySch93}
 {Spirtes, P., Glymour C.  and  Scheines R.} (1993).
\textit{Causation, prediction and search}.
Springer,   New York.
\vspace{-2mm}




\bibitem[Stanghellini and Wermuth(2005)]{StaWer05}
{Stanghellini, E. and  Wermuth, N.}  (2005).
On the identification of path analysis models with one  hidden variable.
\textit{Biometrika} \textbf{92}, 337--350.
\vspace{-2mm}




\bibitem[Studen\'y(2005)]{Stu05}
{Studen\'y, M.}   (2005). \emph{Probabilistic conditional independence structures}.
Sprin\-ger,  London.
\vspace{-2mm}


\bibitem[Sundberg(2010)]{Sund10}
{Sundberg, R.} (2010).
Flat and multimodal likelihoods and model lack of fit in curved exponential families.
\textit{Scand. J. Statist.},  to appear.
\vspace{-2mm}


\bibitem[Tarjan and Yannakakis(1984)]{TarYan84}
{Tarjan, R. E. and Yannakakis, M.} (1984).
Simple linear-time algorithms to test chordality of graphs, test acyclicity of hypergraphs, and selectively reduce acyclic hypergraphs
\textit{SIAM J. Comp} \textbf{13}, 566 -- 579.
\vspace{-2mm}

\bibitem[Verma and Pearl(1990)]{VerPea90}
{Verma, T. and Pearl J.}  (1990).
Equivalence and synthesis of causal models.
Proc. 6th  UAI  conf.  (eds.   	
P.P. Bonissone, 	
M.  Henrion,	
L.N. Kanal 	and
J.F. Lemmer). Elesevier, Amsterdam, 220--227.
\vspace{-2mm}

\bibitem[Wermuth(1976a)]{Wer76a}
{Wermuth, N.} (1976a). Analogies between multiplicative models for
contingency tables and covariance selection.
\textit{Biometrics} \textbf{32}, 95--108.
\vspace{-2mm}

\bibitem[Wermuth(1976b)]{Wer76b}
{Wermuth, N.}  (1976b).
Model search among multiplicative models.
\textit{Biometrics} \textbf {32}, 253--263.
\vspace{-2mm}


\bibitem[Wermuth(1980)]{Wer80}
{Wermuth, N.}  (1980). Linear recursive equations, covariance selection,
and path analysis.
\textit{J. Amer. Statist. Assoc.}
\textbf{75}, 963--97.
\vspace{-2mm}




\bibitem[Wermuth(2011)]{Wer10}
{Wermuth, N.}  (2011).
Probability models with summary graph structure.
\textit{Bernoulli},
\textbf{17},
845-879.
\vspace{-2mm}



\bibitem[Wermuth and Cox(1998)]{WerCox98}
{Wermuth, N. and Cox, D. R.} (1998).
On association models defined over independence graphs.
\textit{Bernoulli}
\textbf{4}, 477--495.


\bibitem[Wermuth and Cox(2004)]{WerCox04}
{Wermuth, N. and Cox, D. R.} (2004).
Joint response graphs and separation induced by triangular
systems.
\textit{J.Roy. Stat. Soc. B}
\textbf{66}, 687-717.
\vspace{-2mm}



\bibitem[Wermuth and  Cox(2008)] {WerCox08}
{Wermuth, N. and Cox, D. R.} (2008).
Distortions of effects caused by indirect confounding.
\textit{Biometrika}
\textbf{95}, 17--33.
\vspace{-2mm}

\bibitem[Wermuth, Cox and Marchetti(2006)]{WerCoxMar06}
{Wermuth, N., Cox, D. R. and  Marchetti, G. M.} (2006).
Covariance chains.
\textit{Bernoulli}
 \textbf{12}, 841-862.
\vspace{-2mm}


 \bibitem[Wermuth and  Lauritzen(1983)] {WerLau83}
 {Wermuth, N. and   Lauritzen, S. L.} (1983). Graphical and recursive
models for contingency tables.
\textit{ Biometrika} \textbf{70}, 537--552.
\vspace{-2mm}

\bibitem[Wermuth and  Lauritzen(1990)]{WerLau90}
{Wermuth, N. and Lauritzen, S. L.} (1990).
 On substantive research
hypotheses, conditional independence graphs and graphical chain models
(with discusssion). \textit{J. Roy. Statist. Soc. B}
\textbf{52}, 21--75.
\vspace{-2mm}


\bibitem[Wermuth, Marchetti and Cox(2009)]{WerMarCox09}
{Wermuth, N.,  Marchetti, G. M.  and Cox, D. R.} (2009).
Triangular systems for symmetric
binary variables.
\textit{Electr. J. Statist.}
\textbf{3}, 932--955.
\vspace{-2mm}




\bibitem[Wermuth, Wiedenbeck and Cox(2006)]{WerWieCox06}
{Wermuth, N.,  Wiedenbeck, M. and Cox, D. R.} (2006).
Partial inversion for linear systems and partial closure of independence
graphs.
\textit{BIT, Numerical Mathematics}
\textbf{46}, 883--901.
\vspace{-2mm}


\bibitem[Whittaker(1990)]{Whit90}
{Whittaker, J.} (1990).
\textit{Graphical models in applied multivariate statistics.} Wiley, Chi\-chester.
\vspace{-2mm}




\bibitem[Wiedenbeck and Wermuth(2010)]{WieWer10}
{Wiedenbeck, M. and Wermuth, N.} (2010).
Changing parameters by partial mappings. \textit{Statistica Sinica}
\textbf{20},  823--836.
\vspace{-2mm}


 \bibitem[Zellner(1962)]{Zell62}
 Zellner, A. (1962). An efficient method of estimating seemingly unrelated regressions and tests for aggregation bias.
\textit{J. Amer. Statist. Assoc.} \textbf{ 57}, 348--368.
\vspace{-2mm}

\bibitem[Zhao, Zheng  and Liu(2005)]{ZhaZheLiu05}
{Zhao, H., Zheng, Z. and Liu, B.}  (2005).
On the Markov equivalence of maximal ancestral graphs.
\textit{Science China  A}
\textbf{48},  548--562.
\vspace{-2mm}



}

\end{thebibliography}
\end{document}